\documentclass[11pt]{article}

\usepackage[margin=20mm]{geometry}

\usepackage{amsmath, amsthm, amssymb, authblk, wasysym, verbatim, bbm, color, graphics, geometry, relsize, setspace,graphicx, subcaption,cleveref,comment}
\usepackage{xcolor}
\usepackage[normalem]{ulem}
\newcommand{\R}{\mathbb{R}}

\newcommand{\stkout}[1]{\ifmmode\text{\sout{\ensuremath{#1}}}\else\sout{#1}\fi}

\newtheorem{thm}{Theorem}
\newtheorem{lem}{Lemma}

\newtheorem{fact}{Fact}
\theoremstyle{definition}

\newtheorem{rem}{Remark}
\numberwithin{equation}{section}

\title{The convexity of optimal transport-based waveform inversion for certain structured velocity models}
\author{Srinath Mahankali\footnote{smahankali10@stuy.edu}\\%
    Stuyvesant High School\\%
    New York, NY 10282
}
\date{
    Faculty Advisor: Dr.\ Yunan Yang\footnote{yunan.yang@nyu.edu}\\%
    Courant Institute of Mathematical Sciences\\%
    New York, NY 10012-1185
}
\begin{document}
\maketitle
\begin{abstract}
Full--waveform inversion (FWI) is a method used to determine properties of the Earth from information on the surface. We use the squared Wasserstein distance (squared $W_2$ distance) as an objective function to invert for the velocity of seismic waves as a function of position in the Earth, and we discuss its convexity with respect to the velocity parameter. 
In one dimension, we consider constant, piecewise increasing, and linearly increasing velocity models as a function of position, and we show the convexity of the squared $W_2$ distance with respect to the velocity parameter on the interval from zero to the true value of the velocity parameter when the source function is a probability measure. Furthermore, we consider a two--dimensional model where velocity is linearly increasing as a function of depth and prove the convexity of the squared $W_2$ distance in the velocity parameter on large regions containing the true value. We discuss the convexity of the squared $W_2$ distance compared with the convexity of the squared $L^2$ norm, and we discuss the relationship between frequency and convexity of these respective distances. We also discuss multiple approaches to optimal transport for non--probability measures by first converting the wave data into probability measures.
\end{abstract}
\section{Introduction}

The study of seismic waves has many practical applications in geology, especially in searching for natural resources such as oil or natural gas.
It plays a major role in determining the type of material underground, given the position of several receivers on the surface and the amount of time it takes for the wave to rebound to the surface. The velocity of the wave (as a function of its position) is unknown, and finding the wave velocity function is equivalent to finding the underground substance. With the same wave source, different wave velocity properties produce different wave data (such as wave amplitude and travel time) measured at a given receiver. This wave data can be used to find the velocity. We use an objective function, or misfit function, which measures the ``distance" between two sets of wave data. This allows us to compare the observed data with simulated data to find the true velocity function. Some examples of objective functions are the $L^p$ distance and the $p$th Wasserstein distance ($W_p$ distance) from the theory of optimal transport~\cite{villani2003topics}, the latter of which is the main tool of this research project.

The objective function becomes zero when the observed data and simulated data are equivalent, which occurs when we have the correct velocity model. Thus, finding the correct velocity model is an optimization problem: minimizing the objective function, which measures the error between simulated and observed wave data. Although it is necessary to have a gradient of zero to minimize the objective function, this is not enough, as it is possible to reach a saddle point or a local minimum at such a point. However, this issue is fixed if the objective function is convex, as it will have only one global minimum. Thus, it is important for the convex region near the global minimum to be as large as possible. For this reason, we investigate the convexity of the squared $W_2$ distance as an objective function.

The conventional choice of objective function is the least--squares $(L^2)$ norm, used in both time~\cite{tarantola2005inverse} and frequency~\cite{pratt1990inverse1, pratt1990inverse2} domains.
If we have data from multiple receiver locations $X_r$ (where $r$ is the index for the receiver location), we can consider the squared $L^2$ norm of the difference between the predicted wave data $g(X_r,t,c)$ and the observed wave data $h(X_r,t) = g(X_r,t,c^*)$ \begin{equation}\label{eq:fullL2norm}
    \mathcal{L}(c) = \frac{1}{2}\sum_r\int_0^\mathcal{T} |g(X_r,t,c)-h(X_r,t)|^2 \hspace{1mm}\text{d}t,
\end{equation} from~\cite{baek2014velocity}, where $c^*$ is the true velocity parameter.
While $\mathcal{L}(c)$ is minimized (and therefore equal to zero) exactly when $c = c^*,$ algorithms for minimizing the squared $L^2$ norm may reach local minima instead of the global minimum when $c = c^*,$ due to the nonconvexity of the squared $L^2$ norm as discussed in~\cite{baek2014velocity,yang2018optimal}. In addition, its sensitivity to noise can make it an unsuitable choice of objective function~\cite{bozdaug2011misfit}. Thus, we use the squared $W_2$ distance from~\cite{chen2018quadratic,yang2018optimal,yang2019analysis} instead. If there are multiple receiver locations $X_r,$ our final objective function will be \begin{equation}\label{eq:fullW2norm}
    \mathcal{W}(c) = \sum_{r} W_2^2(g(X_r,t,c),h(X_r,t)).
\end{equation}
Previous results show that the squared $W_2$ metric is jointly convex in translations and dilations of the data~\cite{yang2019analysis}, suggesting that the squared $W_2$ metric is a suitable choice for the objective function. Taking the function $$f(x) =\begin{cases} 
      \frac{1}{2\pi}\sin^2(x), & -2
      \pi\le x \le 2\pi, \\
      0, & \mathrm{otherwise.}
   \end{cases} $$ as an example, 
we compare in \Cref{fig:Comparison of W2 and L2 Distance} the graphs of the squared $L^2$ norm of $f(x) - f(x-s)$ with the squared $W_2$ distance between $f(x)$ and $f(x-s)$. It can be observed that the squared $W_2$ distance is convex in the shift $s$ while the graph of the squared $L^2$ norm is not.
\begin{figure}
\begin{subfigure}{.5\textwidth}
  \centering
  \includegraphics[width=.8\linewidth]{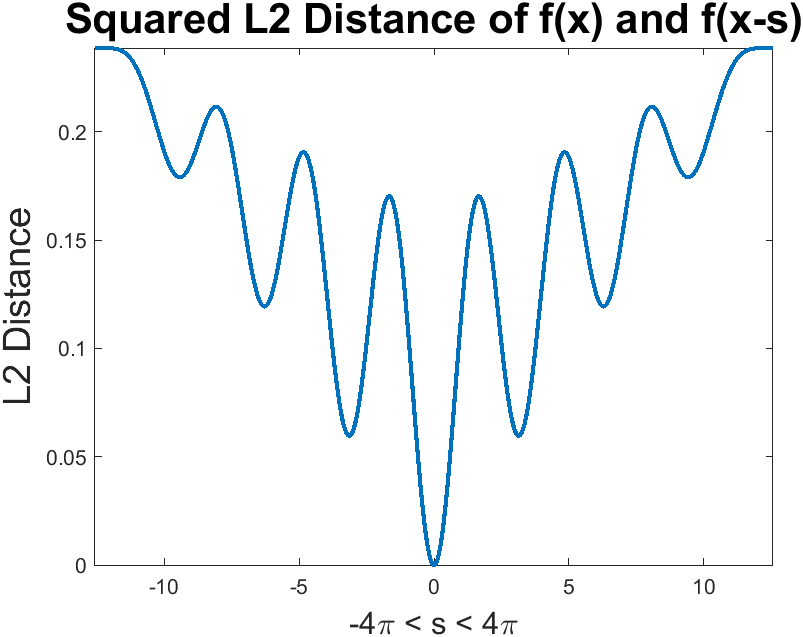}
  \caption{Squared $L^2$ metric as a plot in $s$}
  \label{fig:l2sinplot}
\end{subfigure}%
\begin{subfigure}{.5\textwidth}
  \centering
  \includegraphics[width=.8\linewidth]{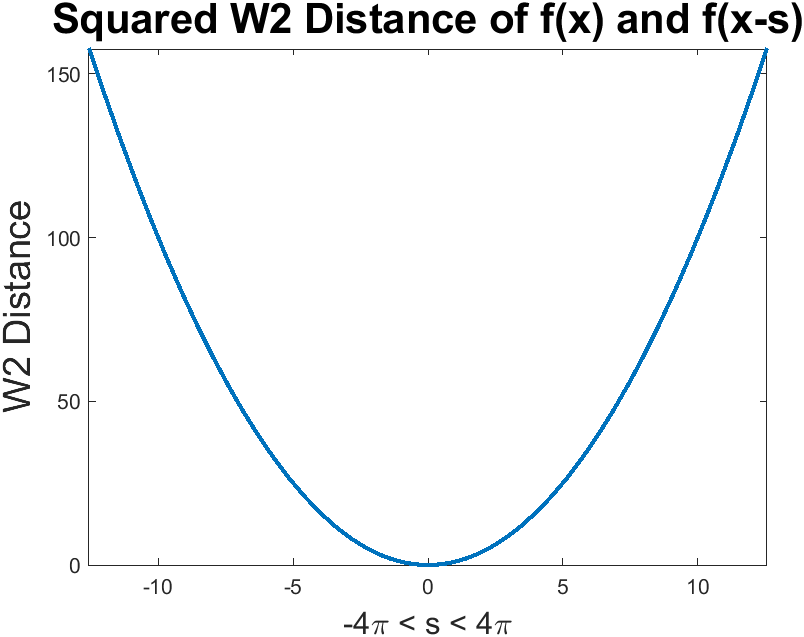}
  \caption{Squared $W_2$ distance as a plot in $s$}
  \label{fig:w2sinplot}
\end{subfigure}%
    \caption{Comparison between squared $W_2$ distance of $f(x)$ and $f(x-s)$ and squared $L^2$ norm of $f(x) - f(x-s)$ as plots in the shift $s$.}
    \label{fig:Comparison of W2 and L2 Distance}
\end{figure}

In this paper, we present a theoretical approach to velocity inversion using optimal transport by investigating the convexity of the squared $W_2$ distance as a function of the velocity parameters -- this has not been studied theoretically before. We investigate several velocity models in one dimension and show that the squared $W_2$ distance is a suitable objective function when inverting for the velocity parameter. In two dimensions, we consider a particular velocity model, and we show that the squared $W_2$ distance is a suitable objective function in the case where the source function $f$ is nonnegative. We generalize to when the source function alternates between negative and positive values, and we show that the squared $W_2$ distance is a suitable objective function given certain requirements on $f$. These theorems suggest that the squared $W_2$ distance is a suitable objective function when $\inf f$ is close enough to zero. Numerical evidence suggests that the squared $W_2$ distance is a better objective function than the squared $L^2$ norm when the source function is nonnegative.

The paper is structured as follows. In Section 2, we discuss background knowledge which is used later on in this paper. First, we introduce optimal transport and the $W_p$ distance, along with an explicit formula as well as some of its properties. Furthermore, we discuss the solution of wave equations in one dimension and introduce \textit{ray tracing}, which is used to solve wave equations in higher dimensions.
In Section 3, we first consider one--dimensional velocity models, where the velocity is either constant, piecewise constant, or linearly increasing as a function of position, and the source function is a probability measure. 
Then, we study a two--dimensional velocity model where the wave velocity $v$ satisfies $v(X,z) = a+bz$ where $a$ and $b$ are positive constants, $X$ is the horizontal position, and $z$ is the depth of the wave. Initially, in \Cref{sssection:constant_amplitude}, we assume that the source function $f$ is a probability measure and that the wave amplitude is unchanged, and we show \Cref{thm:2d_time_shift}.
In \Cref{sssection: vary_amplitude}, we involve the wave amplitude and we allow the source function to alternate between positive and negative values, and we show \Cref{thm:W2_amp_convex}.
In Section 4, we compare the convexity of the squared $W_2$ distance with the squared $L^2$ norm, and we include numerical examples. We also discuss the relationship of the $W_2$ distance with the $\dot{\mathcal{H}}^{-1}$ norm, and we discuss multiple approaches to optimal transport for non--probability measures.
We summarize this paper in Section 5 and discuss a possible direction for future research. 

\section{Background}
We introduce optimal transport, and essential background on wave equations. 
\subsection{Optimal Transport}
In this section, we establish the main goal of optimal transport (originally introduced by Monge~\cite{monge1781memoire}) and introduce the $W_p$ distance along with some of its useful properties. Optimal transport involves probability spaces, which are nonnegative measure spaces with total measure equal to one. We discuss the convexity of the squared $W_2$ metric, which is a distance between two probability measures on a probability space, and introduce some of its properties.
\subsubsection{General Problem}
\begin{figure}
    \centering
    \includegraphics[width=.5\linewidth]{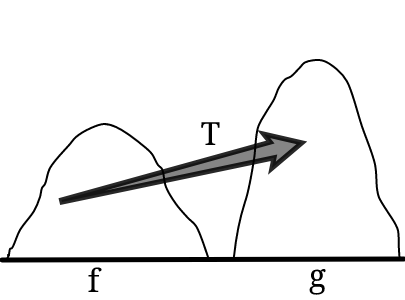}    
    \caption{The two functions $f$ and $g$ have the same total area, and $T$ maps all points of $f$ to $g$. The goal of optimal transport is to find $T$ such that the cost function is minimized.}
    \label{fig:Optimal Transport}
\end{figure}

Consider two distinct probability measures $\mu$ and $\nu$ defined on the Borel sets of $\R^n$.
The goal of optimal transport is to find a map $T:\R^n\to\R^n$, shown in \Cref{fig:Optimal Transport}, which minimizes the total cost of mapping $\mu$ to $\nu$ according to the map $T,$ for a given cost function $c:\R^n\times\R^n\to \R$~\cite{ambrosio2013user}.
The $W_p$ metric, based on optimal transport~\cite{yang2018analysis}, gives the optimal transportation cost when the cost function is $c(x,y) = |x-y|^p$.
\subsubsection{Computing the $W_2$ metric in one dimension} Let $\mu$ and $\nu$ be probability measures defined on the Borel sets of $\R^n$.
We define the $W_p$ distance as 
\begin{equation*}
     W_p(f,g) = \left(\mathop{\inf}_{T\in\mathbb{M}(\mu,\nu)}\int_{\Omega}|x-T(x)|^p\mathrm{d}\mu\right)^{\frac{1}{p}}
\end{equation*}
where $f(x)\text{d}x = \text{d}\mu$, $g(y)\text{d}y = \text{d}\nu$, $\Omega$ is the support of $\mu$, and $\mathbb{M}(\mu,\nu)$ is the set of all mass--preserving maps $T:\R^n\to\R^n$ which map $\mu$ to $\nu$.
The integral within the infimum is the total cost of the transport map $T,$ so computing $W_p(f,g)$ is equivalent to minimizing the transport cost. 
We study the case where $p = 2,$ and our objective function is $W_2^2(f,g)$.

In one dimension, it is possible to express $W_2^2(f,g)$ for probability measures $f$ and $g$ in a simpler way: Let $F$ and $G$ be the cumulative distribution functions of $f$ and $g,$ respectively. 
Then, Rachev and R\"{u}schendorf derive the formula for the squared $W_2$ distance in one dimension~\cite{rachev1998mass} as
\begin{equation}\label{eq:W2_formula}
    W_2^2(f,g) = \int_0^1(F^{-1}(s) - G^{-1}(s))^2 \hspace{1mm}\text{d}s = \int_\Omega (t-G^{-1}(F(t)))^2f(t)\hspace{1mm}\text{d}t
\end{equation} where $\Omega$ is the domain of $f$. 
This formula for $W_2^2(f,g)$ is useful when the wave data is only a function of time.

As wave data are not usually probability density functions, we can normalize a function $k(t)$ defined on $[0,\mathcal{T}]$ by replacing it with $$\frac{k(t) + \gamma}{\int_0^{\mathcal{T}} (k(t) + \gamma)\hspace{1mm}\text{d}t}$$ for some constant $\gamma > 0$ such that $k(t) + \gamma > 0$ for all $t\in[0,\mathcal{T})$~\cite{engquist2018seismic}. 
\subsubsection{Computing the $W_2$ Metric in Higher Dimensions}
In general, there is no explicit formula to compute the $W_2$ metric in higher dimensions. However, certain requirements derived from the concept of cyclical monotonicity~\cite{knott1984optimal} make it possible to calculate the optimal map $T$, and therefore the $W_2$ metric, through numerical methods. This is shown in the following theorem of Brenier~\cite{brenier1991polar,de2014monge,yang2018optimal}:
\begin{thm}[Brenier's theorem]
Let $\mu$ and $\nu$ be two compactly supported probability measures on $\R^n$. 
If $\mu$ is absolutely continuous with respect to the Lebesgue measure, then there is a convex function $w:\R^n \to \R$ such that the optimal map $T$ for the cost function $c(x,y) = |x-y|^2$ is given by $T(x)=\nabla w(x)$ for $\mu$--almost every $x$.
\end{thm}
Furthermore, if $\mu(\text{d}x) = f(x)\text{d}x, \nu(\text{d}y) = g(y)\text{d}y$, then $T$ is differentiable $\mu$--almost everywhere and
\begin{equation}\label{eq:monge_ampere_1}
    \det(\nabla T(x)) = \frac{f(x)}{g(T(x))},
\end{equation}
from the mass preserving property of $T$.
Replacing $T(x)$ in \Cref{eq:monge_ampere_1} with $\nabla w(x)$ leads to the \textit{Monge--Amp\`{e}re} equation
\begin{equation*}
    \det(D^2w(x)) = \frac{f(x)}{g(\nabla w(x))},
\end{equation*}
where $D^2w(x)$ is the Hessian matrix of $w$. Then, the squared $W_2$ distance satisfies \begin{equation*}
    W_2^2(f,g) = \int_\Omega |x-\nabla w(x)|^2 f(x)\hspace{1mm}\text{d}x,
\end{equation*} where $\Omega$ is the domain of $f$.
\subsubsection{Properties of the Squared $W_2$ Distance} Results about the convexity of the squared $W_2$ distance with respect to changes in the data are known~\cite{engquist2016optimal,yang2018optimal,yang2019analysis}.
\begin{thm}\label{thm: shift_dilation_w2}
Let $f$ and $g$ be compactly supported probability density functions on an interval $\Omega\subset\R$. Then,
\begin{equation}\label{eq:W2_shift_formula}
    W_2^2(f(t-s),g(t)) = W_2^2(g(t),f(t)) + s^2 + 2s\int_\Omega (x-T(x))f(x)\hspace{1mm}\mathrm{d}x
\end{equation} where $s\in\R$ and $T$ is the optimal map from $f$ to $g$. Furthermore, $W_2^2(f(t),Af(At-s))$ is convex in both $A$ and $s$, for $A \in \R^+$.
\end{thm}
The convexity with respect to shifts and dilations suggests that the squared $W_2$ distance is more suitable when inverting for wave data.
\subsection{Background on Wave Equation}
We introduce the partial differential equation which governs the behavior of $n$--dimensional waves. We also introduce d'Alembert's solution to the one--dimensional wave equation and present the \textit{ray tracing} approach to solving higher dimensional wave equations.

\subsubsection{General Wave Equation}
 An $n$--dimensional wave can be expressed as a function $\psi$ of $n$ position variables $x_1, x_2, \ldots, x_n$ and time $t$, which in general satisfies the partial differential equation \begin{equation*}
    \frac{\partial^2\psi}{\partial t^2} - \mathcal{C}(\mathbf{x})^2\left(\frac{\partial^2\psi}{\partial {x_1}^2} + \frac{\partial^2\psi}{\partial {x_2}^2} + \cdots + \frac{\partial^2\psi}{\partial {x_n}^2}\right) = f(\mathbf{x},t)
\end{equation*}
for a variable coefficient $\mathcal{C}(\mathbf{x})$ which is a function defined on $\R^n$, and a source function $f$ which is a function of both space and time. In general, the wave equation might not have an analytical solution.
\subsubsection{Solution to 1D Wave Equation}
In one dimension, we consider a simple case where $\mathcal{C}(\mathbf{x})$ is equal to a constant $c$. The partial differential equation becomes \begin{equation*}
    \frac{\partial^2\psi}{\partial t^2} = c^2\frac{\partial^2\psi}{\partial x^2},
\end{equation*}
but unlike its $n$--dimensional variant, it is possible to obtain an explicit solution as derived by d'Alembert~\cite{demanet2016waves}: 
\begin{equation}
    \psi(x,t;c) = \frac{j(x+ct) + j(x-ct)}{2} + \frac{1}{2c}\int_{x-ct}^{x+ct} k(s) ds \label{eq:d'Alembert_solution}
\end{equation}
given the initial conditions $\psi(x,0) = j(x)$ and $\psi_t(x,0) = k(x)$. 
\subsubsection{Solution for Higher Dimensional Wave Equation}
\begin{figure}
    \centering
    \includegraphics[scale = 0.55]{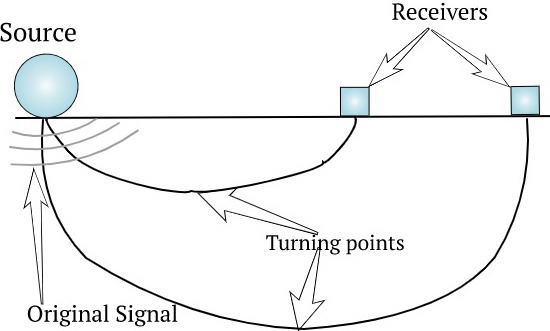}
    \caption{The method of \textit{ray tracing} tracks the path of a ray passing through the Earth, and the paths of two rays are shown here. Each ray leaves the source location, reaches a turning point, and then returns to the surface at the indicated receiver location.}
    \label{fig:raytrace1}
\end{figure}
In higher dimensions, however, a wave equation might not have an analytical solution, so we use \textit{ray tracing} to obtain information about the wave function~\cite{shearer2019introduction}. An example of ray tracing is shown in \Cref{fig:raytrace1}, where an original signal is released, and two receivers on the surface are present to measure the wave data. We assume that the velocity $v$ of a wave can be expressed as a function of depth $z$ in the Earth, and the ray parameter $p,$ or the horizontal slowness, can be expressed as $u(z)\sin\theta$ by Snell's law where $u(z) = \frac{1}{v(z)}$ and $\theta$ is the angle the ray makes with a vertical axis. The ray parameter is constant throughout the path of the ray. We also define the vertical slowness $\eta(z)$ as $\sqrt{u(z)^2 - p^2}$.
The path of the ray is symmetric about a vertical line passing through a turning point at depth $z_p$, and $u(z_p) = p$ while $\eta(z_p) = 0$.
By examining a ray passing through several layers in the Earth and eventually returning to the surface, we may calculate the horizontal distance and the traveltime of the wave as a function of the velocity. We can calculate \begin{equation} \label{distance}
    X = 2p\int_0^{z_p} \frac{dz}{\sqrt{u^2(z) - p^2}} 
\end{equation} as the distance from the source to the receiver, and \begin{equation}\label{traveltime}
    T = 2\int_0^{z_p} \frac{u^2(z)}{\sqrt{u^2(z) - p^2}} dz
\end{equation} as the total travel time. Using \Cref{distance}, it is possible to solve for $p$ in terms of $X$ and the velocity parameters. It is also possible to use \Cref{traveltime} to solve for $T$ as a function of $X$ and the velocity parameters, and it is possible to  calculate the predicted time $T_{pred}$ by guessing the velocity parameter. The method of \textit{traveltime tomography} uses the above formulas to calculate the predicted time and approaches the problem as minimizing the squared difference between the predicted time and the observed time: $(T_{pred} - T_{obs})^2$~\cite{zelt2011traveltime}.

However, this method does not work when the velocity model is not continuous~\cite{shearer2019introduction}. In addition to using \Cref{distance,traveltime}, we make use of the wave's amplitude as well to deal with discontinuous velocity models and other issues that traveltime tomography runs into. This method is known as \textit{full--waveform inversion} (FWI), where both the amplitude and the traveltime are used to approximate the properties of the Earth~\cite{virieux2009overview}.

The final amplitude is asymptotically scaled by a factor of \begin{equation}\label{amplitude}
    A = \left(2\int_0^{z_p}\frac{u(z)}{\sqrt{u^2(z) - p^2}}dz\right)^{-1},
\end{equation} which is the reciprocal of the total arc length of the ray's path~\cite{shearer2019introduction}, and the ray's path is symmetric about the turning point at depth $z_p$. Thus, if the source function is $f(t)$, then the observed wave function is approximated by $Af(t - T_{obs})$. 
To invert for the velocity we use an objective function such as  \Cref{eq:fullL2norm} and \Cref{eq:fullW2norm} to compare observed data with simulated data and minimize it using standard algorithms, and we will see that the squared $W_2$ metric can be a suitable choice. 
\section{Convexity in the Model Parameter}
In this section, we study multiple velocity models in one dimension as well as a model in two dimensions, and we prove convexity of the squared $W_2$ distance with respect to the velocity parameter on certain domains. 
First, we consider one--dimensional velocity models. We begin with a model with constant velocity $c$ and prove convexity of the squared $W_2$ distance with respect to $c$. We then consider two models with piecewise increasing velocities with respect to distance from the source and a model where velocity is linearly increasing. In every one--dimensional velocity model, we assume that the source function is nonnegative. 
Finally, we consider a two--dimensional model where the velocity $v$ satisfies $v(X,z) = a+bz$ where $a$ and $b$ are positive constants, $X$ is the horizontal position of the ray, and $z$ is the current depth of the ray. In \Cref{sssection:constant_amplitude}, we assume that the source function is nonnegative, which allows us to use \Cref{thm: shift_dilation_w2} when computing the squared $W_2$ distance. After this, we consider a more general case where the source function is alternating in \Cref{sssection: vary_amplitude}, and the predicted wave function has an amplitude which is a function of $a,b$, and the receiver location $X$. We also use the following result~\cite{287725}:
\begin{lem}\label{lem:compconvex}
Let $P:\Omega_1\to\R$ and $Q:\Omega_2\to\Omega_1$ be convex functions where $\Omega_1\subseteq\R$ and $\Omega_2\subseteq\R^n$ are convex sets and $n \ge 1$. Furthermore, assume $P$ is nondecreasing. Then, $P(Q(\mathbf{x}))$ is a convex function on $\Omega_2$.
\end{lem} 
An example of this is when $P(x) = x^2$. Then, if $Q$ is convex and nonnegative, we see that $Q^2$ is also convex on its domain.
\subsection{Constant Velocity in One Dimension}\label{ssection: 1d constant velocity}
Let $c$ be the constant wave velocity, and let our initial wave function be $f(t)$. For this section, we assume that $f(t)$ is a probability distribution. We assume that the final wave function at a fixed spatial location $d$ is of the form $f(t - T(c)),$ where $T(c)$ is the amount of time it takes to receive the wave signals at location $d$ as a function of the velocity $c$. Since the total distance is $d, T(c) = \frac{d}{c},$ which means the predicted wave function is $f(t - \frac{d}{c})$. Letting $c^*$ be the true value of the velocity gives us \begin{equation*}
    W_2^2\left(f\left(t-\frac{d}{c}\right),f\left(t-\frac{d}{c^*}\right)\right) = \left(\frac{d}{c} - \frac{d}{c^*}\right)^2,
\end{equation*} as the squared $W_2$ distance from \Cref{thm: shift_dilation_w2}, because the integral in \Cref{eq:W2_shift_formula} is zero. In addition, $\frac{d}{c} - \frac{d}{c^*}$ is a convex function of $c$ on $(0,\infty),$ and it is nonnegative on the interval $(0,c^*]$. Hence, by  \Cref{lem:compconvex}, the squared $W_2$ distance is convex on $(0,c^*]$.
\begin{rem}
Due to the convexity of the squared $W_2$ distance in the interval $(0,k^*]$, choosing a small value of $k$ as the initial guess guarantees being able to find the true velocity parameter $k^*$ through gradient--based optimization methods.
\end{rem}
\subsection{Non--constant Velocity in One Dimension}\label{ssection:1d nonconstant velocity}
When the velocity is non--constant, the d'Alembert solution \Cref{eq:d'Alembert_solution} does not hold anymore. Thus,
we study the convexity of the squared $W_2$ distance for several non--constant velocity models, shown in \Cref{fig:Velocity Model Graphs}. We first study a model where the velocity is piecewise constant, after which we add multiple pieces. Then, we study a model where the velocity is linearly increasing as a function of position. For all of these models, we assume that the source function $f(t)$ is a probability distribution.
\begin{figure}
\begin{subfigure}{.3\textwidth}
  \centering
  \includegraphics[width=1.0\textwidth]{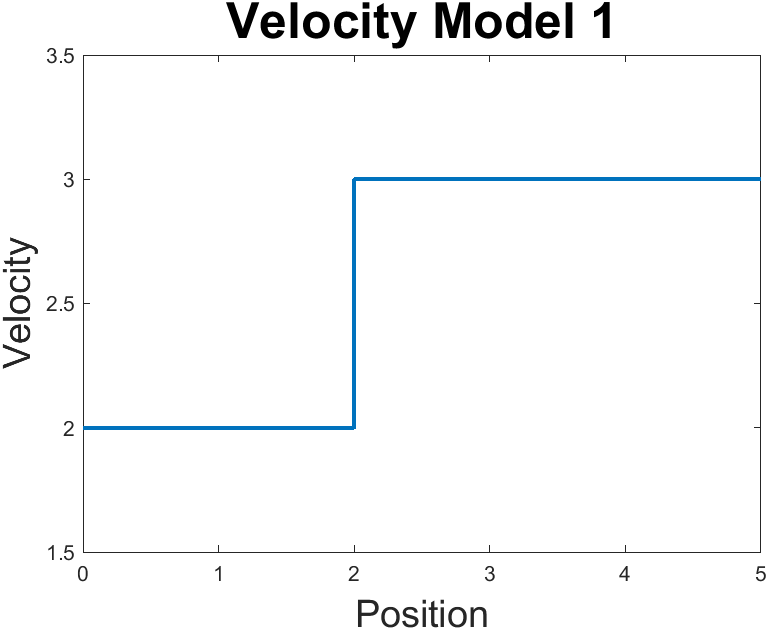}
  \caption{First velocity model.}
  \label{fig:velocitymodel1}
\end{subfigure}%
\begin{subfigure}{.3\textwidth}
  \centering
  \includegraphics[width=1.0\textwidth]{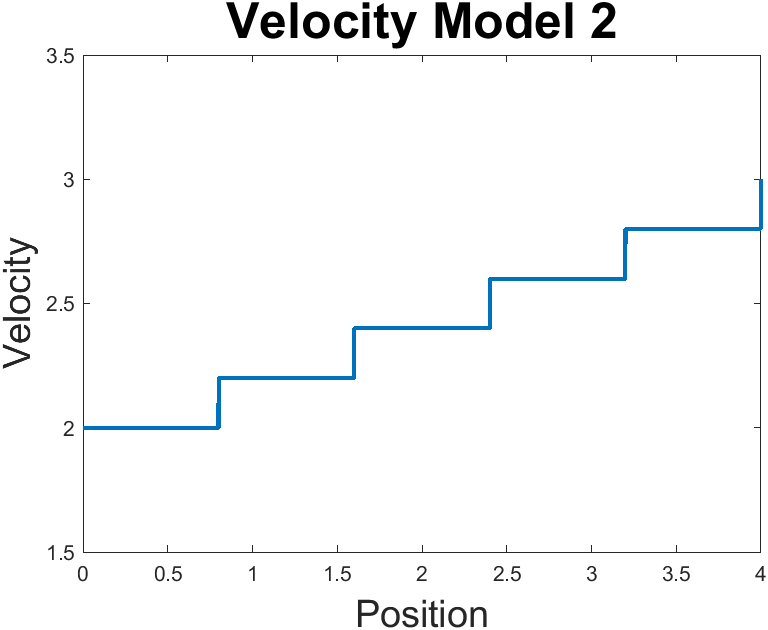}
  \caption{Second velocity model.}
  \label{fig:velocitymodel2}
\end{subfigure}%
\centering
\begin{subfigure}{.3\textwidth}
  \centering
  \includegraphics[width=1.0\textwidth]{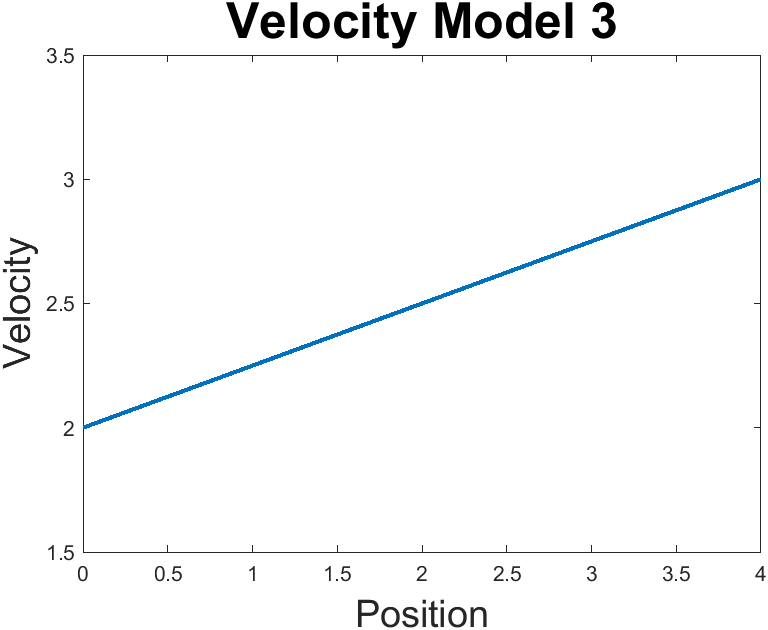}
  \caption{Third velocity model.}
  \label{fig:velocitymodel3}
 \end{subfigure}%
    \caption{Three types of non--constant velocity models in one dimension that we study.}
    \label{fig:Velocity Model Graphs}
\end{figure}
\subsubsection{Piecewise Constant Velocity} 
The first scenario we study is a piecewise constant velocity model.
Assume the velocity $v(x)$ satisfies $v(x) = c_1$ for $0\le x \le d_1$ and a known constant $c_1$, and $v(x) = c_2$ for $d_1 < x \le d_1 + d_2$, as seen in \Cref{fig:velocitymodel1} -- we show convexity in the unknown $c_2$. The total travel time, as a function of $c_2,$ is $T(c_2) = \frac{d_1}{c_1} + \frac{d_2}{c_2}$.
Letting $c_2^*$ be the true value of $c_2$ gives $\left(\frac{d_2}{c_2} - \frac{d_2}{c_2^*}\right)^2$ as the squared $W_2$ distance by \Cref{thm: shift_dilation_w2}. As the function $\frac{d_2}{c_2} - \frac{d_2}{c_2^*}$ is convex and nonnegative on $(0,c_2^*],$ the squared $W_2$ distance is also convex on $(0,c_2^*]$ by \Cref{lem:compconvex}.
\subsubsection{Piecewise Constant Velocity with Multiple Pieces}
The second scenario we consider is the piecewise constant velocity with $n$ pieces of equal length $d,$ such that the velocity of the wave is
\begin{equation*}
v(x) =
    \begin{cases}
 c_1, & x \le d, \\ 
 c_1 + k, & d<x\le 2d, \\ 
 \vdots & \vdots \\
 c_1 + (n-1)k, & (n-1)d < x \le nd. \\
 \end{cases}
\end{equation*} Here, the unknown variable is $k$. In this case, the total travel time is \begin{equation*}
    T(k) = \sum_{m=0}^{n-1} \frac{d}{c_1 + mk},
\end{equation*}
so the squared $W_2$ distance becomes $(T(k) - T(k^*))^2$ by \Cref{thm: shift_dilation_w2} where $k^*$ is the true value of the velocity parameter. We claim that the squared $W_2$ distance is convex in $k$ on the interval $[0,k^*]$. To do this, we initially show that $T(k)$ is convex in $k$ on $[0,\infty)$. Taking the second derivative, we get \begin{equation*}
    T''(k) = \sum_{m=0}^{n-1} \frac{2dm^2}{(c_1+mk)^3},
\end{equation*} which is nonnegative for all $k\ge 0$. Hence, $T(k)$ is convex in $k$ on $[0,\infty)$. Since $T(k)$ is strictly decreasing on $[0,\infty),$ the function $T(k) - T(k^*)$ is nonnegative (and also convex) on the interval $[0,k^*]$. Thus, the squared $W_2$ distance $(T(k) - T(k^*))^2$ is convex in $k$ on the interval $[0,k^*]$ by \Cref{lem:compconvex}.
\subsubsection{Linearly Increasing Velocity}
Next, we consider a linearly increasing velocity, which is of the form $c_1 + kx$ at a position $x,$ for a known constant $c_1$. Here, $k$ is unknown and $k^*$ is the true value of the velocity parameter. Letting $d$ be the total travel distance, we have the following integral representation for the traveltime: \begin{equation}\label{eq:1d_linear_time_integral}
    T(k) = \int_0^d \frac{\text{d}x}{c_1 + kx} = \frac{\ln(c_1 + kd) - \ln(c_1)}{k}.
\end{equation}
We first prove a lemma.
\begin{lem}\label{lem_linear_velocity}
$T(k)$ is convex in $k$ on the interval $[0,\infty)$, where $T(0) = \frac{d}{c_1}$.
\end{lem}
\begin{proof}
We take the second derivative of $T$:
\begin{equation*}
    T''(k) = \frac{\text{d}^2}{\text{d}k^2}\int_0^d \frac{\text{d}x}{c_1 + kx} = \int_0^d \frac{\text{d}^2}{\text{d}k^2}\left(\frac{1}{c_1 + kx}\right)\hspace{1mm}\text{d}x,
\end{equation*} and the second equation follows by the Leibniz Integral Rule. Because $\frac{\text{d}^2}{\text{d}k^2}\left(\frac{1}{c_1 + kx}\right) = \frac{2x^2}{(c_1+kx)^3} \ge 0$, the integrand is nonnegative. This implies the convexity of $T(k)$.
\end{proof}
Now we are ready to show that $(T(k) - T(k^*))^2$ is convex on the interval $[0,k^*]$.
\begin{thm}\label{Theorem Linear Velocity W2 Convexity}
$(T(k) - T(k^*))^2$ is convex in $k$ on the interval $[0,k^*]$, where $T(0) = \frac{d}{c_1}.$
\end{thm}
\begin{proof}
By \Cref{lem_linear_velocity}, $T(k) - T(k^*)$ is convex. As the square of a nonnegative convex function is convex by \Cref{lem:compconvex}, it is enough to determine the interval in which $T(k) - T(k^*)$ is nonnegative. From the integral representation of $T(k)$ in \Cref{eq:1d_linear_time_integral} we see that $T(k)$ is strictly decreasing in $k$. Hence, $T(k) - T(k^*)$ is nonnegative on the interval $[0,k^*]$ implying the convexity of $(T(k) - T(k^*))^2$ on this interval by \Cref{lem:compconvex}.
\end{proof}
\subsection{A Velocity Model in Two Dimensions}\label{subsection:2d_velocity_model}
Consider a velocity model in two dimensions where the predicted velocity at a point $(X,z)$ is of the form $v(X,z) = a+bz$ where $a$ and $b$ are positive constants, $X$ is the horizontal position of the ray, and $z$ is the current depth of the ray. We analyze the travel time of a ray emanating from a single receiver which has a final distance $X$ from the source, as well as the amplitude of the wave at a receiver. We compute the squared $W_2$ distance of the predicted wave function with the observed wave function (with velocity $a^* + b^*z$ at depth $z$) and aim to find a region in $\R^2$ for which this distance is convex in $(a,b).$
\subsubsection{Constant Amplitude}\label{sssection:constant_amplitude}

While the source function and the observed wave function always have the same amplitude in one dimension, for higher dimensions this is not the case. With a source function $f(t)$, the observed wave data is of the form $Af(t - T)$ for constants $A$ and $T$ because there are no reflections when the velocity is a continuous function of depth. However, we analyze the convexity with the assumption that $A = 1$. We treat the observed wave data as $f(t - T(X,a^*,b^*))$ where $T(X,a,b) = T_{pred}$ is the predicted traveltime expressed as a function of $a,b,$ and $X$ and $T(X,a^*,b^*) = T_{obs}$ is the observed traveltime of the wave data. $X$ is the receiver location, or the distance from the source to the receiver. Furthermore, we assume that $f(t)$ is a probability distribution with compact support in the interval $[p_1,p_2]\subset[0,\mathcal{T}]$, where $\mathcal{T} > T_{obs} + p_2$. When the source function is a probability distribution, this method is equivalent to traveltime tomography because the squared $W_2$ distance is equal to $(T_{pred} - T_{obs})^2$ by \Cref{thm: shift_dilation_w2}. 

Even when the observed wave data is not a probability distribution, we may normalize the data to reduce it to this case. For example, if we normalize a nonnegative function $k(t)$ of the form $Af(t-T)$ with compact support in the interval $[0,\mathcal{T}]$ using the formula
$\widetilde{k}(t) = \frac{k(t)}{\int_0^\mathcal{T} k(t)\hspace{1mm}\text{d}t}$, where $\widetilde{k}(t)$ is the normalized wave data, then the amplitude of the normalized data remains constant. Thus, the squared $W_2$ distance becomes $(T_{pred} - T_{obs})^2$.
\begin{figure}
\begin{subfigure}{.5\textwidth}
  \centering
  \includegraphics[scale = 0.5]{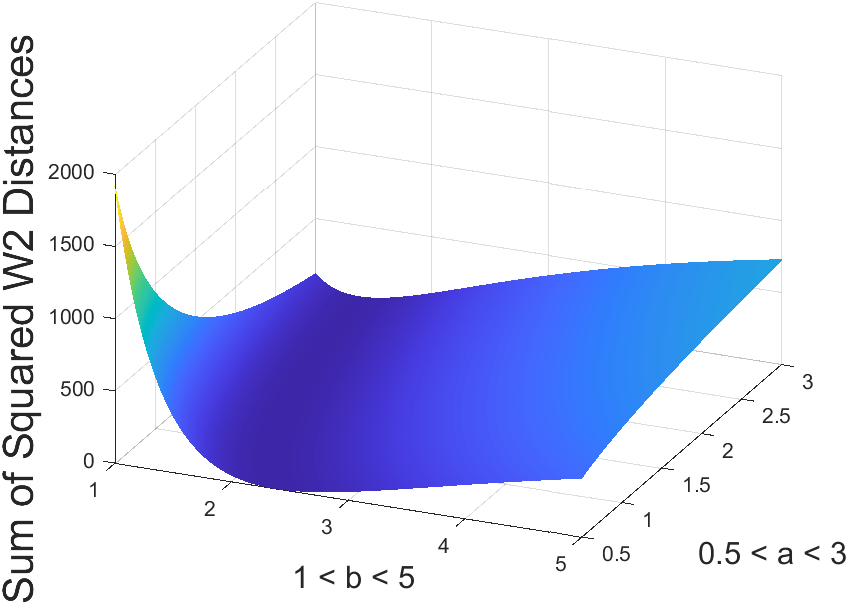}
    \caption{Sum of squared $W_2$ distance}
    \label{fig:W2noamplitude}
\end{subfigure}%
\begin{subfigure}{.5\textwidth}
  \centering
  \includegraphics[scale = 0.5]{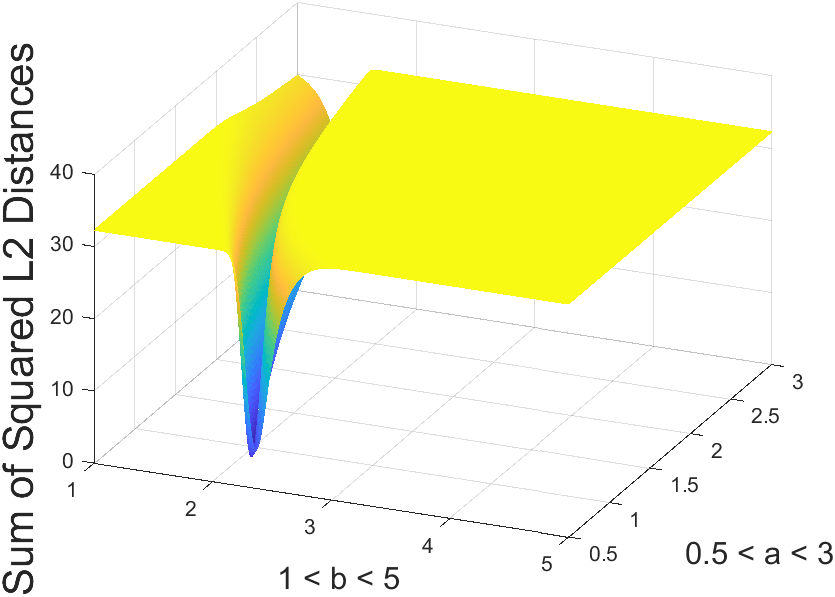}
  \caption{Sum of squared $L^2$ norm}
  \label{fig:L2noampsum}
\end{subfigure}%
\caption{Here, $f(t) = e^{-50(t-5)^2}$ is the source function and $(a^*,b^*) = (1,2)$ is the true value of the velocity parameter. As $f$ is not a probability density function, the wave data, which is of the form $g(X,t,a,b) = f(t - T(X,a,b))$, is normalized according to the formula $\displaystyle{\widetilde{g}(X,t,a,b)= \frac{g(X,t,a,b)}{\int_0^\mathcal{T}g(X,t,a,b)\hspace{1mm}\text{d}t}}$ where $\mathcal{T} = 50$. The objective function
$\mathcal{W}(a,b)=\displaystyle{\sum_{X_r = 10}^{100}W_2^2(\widetilde{g}(t - T(X_r,a,b)),\widetilde{g}(t - T(X_r,a^*,b^*)))}$, where $X_r$ is the location of receiver $r$, is compared with $\mathcal{L}(a,b) = \displaystyle{\sum_{X_r = 10}^{100}\int_0^\mathcal{T} |f(t - T(X_{r},a,b))-f(t - T(X_{r},a^*,b^*))|^2 \hspace{1mm}\text{d}t}$.}
\end{figure}
Using this expression for the squared $W_2$ distance, we can determine its convexity in a large region containing $(a^*,b^*)$ and points arbitrarily close to the origin.

We first explicitly compute $T(X,a,b)$. 
Letting the ray parameter be $p,$ we find that $X = \frac{2\eta_1}{bu_1p}$ from \Cref{distance} where $u_1 = \frac{1}{a}$ is the initial slowness, and $\eta_1 = \sqrt{u_1^2 - p^2}$ is the initial vertical slowness. We can solve for $p$ and $\eta_1$ in terms of $X,a,$ and $b$ as 
\begin{equation*}
    p = \frac{2}{\sqrt{b^2X^2 + 4a^2}},\hspace{5mm} \eta_1 = \frac{bX}{a\sqrt{b^2X^2 + 4a^2}}.
\end{equation*}
From these formulas as well as \Cref{traveltime}, we can solve for the time: \begin{equation*}
    T(X,a,b) = \frac{2}{b}\left(\ln\left(\frac{\sqrt{b^2X^2 + 4a^2} + bX}{2a}\right)\right).
\end{equation*} Then, the squared $W_2$ distance becomes $(T(X,a,b) - T(X,a^*,b^*))^2$ where $v_1(X,z) = a^* + b^*z$ is the true velocity function. We first claim that $T(X,a,b)$ is convex subject to a restriction on $\frac{bX}{2a}$.
\begin{lem}\label{lem_linear_depth}
Let $S_0 > 0$ be the largest root of the equation $S^2\phi(S) + 2\phi(S) - S^2 = 0$, where \begin{equation*}
    \phi(S) = 2\left(1 + \frac{1}{S^2}\right)^{\frac{3}{2}}\ln(\sqrt{S^2 + 1} + S) - \frac{2}{S^2} - 3.
\end{equation*}
Then, the traveltime $T(X,a,b)$ is jointly convex in $(a,b)$ whenever $\frac{bX}{2a} \ge S_0$.
\end{lem}
\begin{proof}
Let $y_1 = \sqrt{b^2X^2 + 4a^2}$ and $S = \frac{bX}{2a}$.
We can compute the first order derivatives of $T(X,a,b)$ as
\begin{equation*}
    \frac{\partial T}{\partial a} = \frac{-2X}{ay_1}\hspace{2mm}\text{and}\hspace{2mm}
    \frac{\partial T}{\partial b} = \frac{2X}{by_1} - \frac{2\ln\left(\sqrt{S^2+1}+S\right)}{b^2}.
\end{equation*}
The Hessian matrix of $T(X,a,b)$, where we treat $X$ as a constant, becomes $\frac{2X}{y_1^3}\mathbf{H}$, where
\begin{equation*}
  \mathbf{H} =  \begin{bmatrix}
  8 + 4S^2 & 2SX \\
  2SX & X^2\left(\frac{2(S^2+1)^{\frac{3}{2}}}{S^3}\ln(\sqrt{S^2 + 1} + S) - \frac{2}{S^2} - 3\right)\\
\end{bmatrix}.
\end{equation*}
To prove the convexity of $T(X,a,b)$, it is enough to show that the Hessian matrix of $T(X,a,b)$ is positive semidefinite, which is equivalent to showing that $\mathbf{H}$ is positive semidefinite.
Using Sylvester's Criterion, because $8+4S^2$ is always positive, we see that $\mathbf{H}$ is positive semidefinite exactly when $\det\mathbf{H} \ge 0.$ Letting \begin{equation*}
    \phi(S) = 2(1 + \frac{1}{S^2})^{\frac{3}{2}}\ln(\sqrt{S^2 + 1} + S) - \frac{2}{S^2} - 3,
\end{equation*}
we see that \begin{equation*}
    \det\mathbf{H} = 8X^2\phi(S) + 4S^2X^2(\phi(S) - 1) = 4X^2(S^2\phi(S) + 2\phi(S) - S^2),
\end{equation*}
so it is enough to show that $S^2\phi(S) + 2\phi(S) - S^2 \ge 0$ for all $S \ge S_0$. Since $S_0$ is the largest root of $S^2\phi(S) + 2\phi(S) - S^2 = 0$, it is enough to show this inequality for all sufficiently large $S$. Observe that $\phi(S) \ge \ln(S)$ for all sufficiently large $S$, implying that $S^2\phi(S) + 2\phi(S) - S^2 \ge S^2(\ln(S) - 1) + 2\ln(S),$ which is at least $0$ for all sufficiently large $S$. Since $S = \frac{bX}{2a}$, this proves the lemma.
\end{proof}
We are now ready to prove the convexity of $(T(X,a,b) - T(X,a^*,b^*))^2$ over a certain region $U$.
\begin{thm}\label{thm:2d_time_shift}
Let $a^*$ and $b^*$ be positive constants, and let
\begin{equation*}
    U := \{(a,b)\in\R^2: a,b > 0, \frac{bX}{2a}\ge S_0, T(X,a,b) \ge T(X,a^*,b^*) \},
\end{equation*}
where $X > 0$ is a fixed constant. Then $U$ is nonempty and $(T(X,a,b) - T(X,a^*,b^*))^2$ is jointly convex in $(a,b)\in U$.
\end{thm}
\begin{proof}
First, as $a$ decreases, both $\frac{bX}{2a}$ and $T(X,a,b)$ increase and approach infinity, implying that $U$ is nonempty. From \Cref{lem_linear_depth} we see that the function $T(X,a,b) - T(X,a^*,b^*)$ is jointly convex in $(a,b)\in U$. In addition, as $T(X,a,b) \ge T(X,a^*,b^*)$ on $U$ we have that $T(X,a,b) - T(X,a^*,b^*)$ is also nonnegative on $U$. Thus, by \Cref{lem:compconvex}, the squared $W_2$ distance, which is $(T(X,a,b) - T(X,a^*,b^*))^2$, is jointly convex in $(a,b)\in U$.
\end{proof}
\begin{rem}
To ensure that $(a^*,b^*)\in U$, it is enough to require that $\frac{b^*X}{2a^*} \ge S_0$ because $T(X,a^*,b^*) - T(X,a^*,b^*)$ is equal to $0$. Thus,
choosing large values of $X$ will ensure that the squared $W_2$ distance is convex in $(a,b)$ in a region containing $(a^*,b^*)$. To find a suitable initial guess, we may choose a point $(a,b)$ such that $\frac{bX}{2a} \ge S_0$ and scale it by a sufficiently small constant in order to satisfy the condition $T(X,a,b) \ge T(X,a^*,b^*)$.
\end{rem}
However, this function is not suitable as an objective function because $(T(X,a,b) - T(X,a^*,b^*))^2$ may equal $0,$ its minimum, even when $a$ is not equal to $a^*$ or $b$ is not equal to $b^*$ -- there is not enough information to find $a^*$ and $b^*$ through one receiver alone. To fix this, we add multiple receiver locations $X_r$. 
\begin{fact}
The equation $T(X,a_1,b_1) - T(X,a_2,b_2) = 0$ has at most one solution in $X > 0$ where $(a_1,b_1)\neq (a_2,b_2)$ are fixed ordered pairs of positive real numbers.
\end{fact}

\begin{proof}
Assume for the sake of contradiction that the equation $T(X,a_1,b_1) = T(X,a_2,b_2)$ has two positive solutions.
Since $T(0,a_1,b_1) = T(0,a_2,b_2) = 0$, there are at least three nonnegative solutions in $X$ to the equation $T(X,a_1,b_1) = T(X,a_2,b_2)$. Therefore, the equation $T'(X,a_1,b_1) = T'(X,a_2,b_2)$, or equivalently, $$\dfrac{2}{\sqrt{b_1^2X^2+4a_1^2}}=\dfrac{2}{\sqrt{b_2^2X^2+4a_2^2}},$$ has at least two positive solutions in $X$ by the Mean Value Theorem. Thus, there are two positive solutions to the equation $(b_1^2-b_2^2)X^2 + (4a_1^2 - 4a_2^2) = 0$, a contradiction.
\end{proof}
This implies that the value of $T(X_r,a,b)$ at two different receiver locations $X_r$ uniquely determines the pair $(a,b)$. We use the objective function $\sum_r(T(X_r,a,b) - T(X_r,a^*,b^*))^2$,
where the sum is over the receivers $r$, and the $X_r$ are the receiver locations. We plot this objective function in \Cref{fig:W2noamplitude}, where $a^* = 1$ and $b^* = 2$ and the $X_r$ range from $10$ to $100$ inclusive. We compare it to the sum of the squared $L^2$ norm where the set of receiver locations is the same and the source function is $f(t) = e^{-50(t-5)^2}$. The $L^2$--based objective function is mostly flat with a sharp incline close to the true velocity parameter $(1,2)$. It is clearly nonconvex, as shown in \Cref{fig:L2noampsum}. On the other hand, the $W_2$--based objective function (which does not depend on the source function, as long as it is nonnegative) appears convex in $(a,b)$.
\subsubsection{Varying Amplitude}\label{sssection: vary_amplitude}
In general, the amplitude of the wave equation solution is non--constant. The predicted wave function at receiver $X$ can be considered to be of the form 
\begin{equation*}
    g(t,a,b) = A(X,a,b)f(t - T(X,a,b)),
\end{equation*} 
where $A(X,a,b)$ is the amplitude as a function of the receiver location $X$ and the velocity parameters $a,b$. The observed wave function at receiver $X$ can be considered of the form $g(t,a^*,b^*)$ where $(a^*,b^*)$ is the true velocity parameter.
Using \Cref{amplitude}, we obtain 
\begin{equation*}
    A(X,a,b) = \frac{b}{\sqrt{b^2X^2 + 4a^2}\left(\frac{\pi}{2} - \arcsin\left(\frac{2a}{\sqrt{b^2X^2 + 4a^2}}\right)\right)}.
\end{equation*} 
Since $b$ and $X$ are positive, we may simplify this expression to get \begin{equation}\label{eq:amp_simplified}
A(X,a,b) = \frac{1}{X\sqrt{1 + \left(\frac{2a}{bX}\right)^2}\left(\frac{\pi}{2} - \arcsin\left(\frac{1}{\sqrt{\left(\frac{bX}{2a}\right)^2 + 1}}\right)\right)} = \frac{1}{RX\sqrt{1 + \left(\frac{1}{S}\right)^2}}
\end{equation} where 
\begin{equation*}
    S = \frac{bX}{2a}\hspace{2mm}\text{and}\hspace{2mm}
    R = \frac{\pi}{2} - \arcsin\left(\frac{1}{\sqrt{S^2 + 1}}\right).
\end{equation*} While the amplitude function in \Cref{eq:amp_simplified} is not fully accurate, it is still a good approximation if the velocity model is continuous. 

We will only consider points $(a,b)$ with $a > \frac{a^*}{100}$ and $\frac{b}{a} \ge \frac{b^*}{a^*}$, which is a positive constant independent of $X$. This also implies that $b > \frac{b^*}{100}$. We treat $X$ as a sufficiently large constant, which causes $S = \frac{bX}{2a} \ge \frac{b^*X}{2a^*}$ to be large as well. In particular, $X$ is taken to be large enough so that $\frac{b^*X}{2a^*} > S_0$.
For large $S$, the amplitude is approximately $\frac{2}{\pi X}$. Taking the derivative of the amplitude with respect to $S$ gives \begin{equation*}
\frac{\partial}{\partial S}A(X,S) =
     \frac{1}{X}\left(\dfrac{1}{\left(\frac{1}{S^2}+1\right)^\frac{3}{2}S^3R}-\dfrac{S}{\left(\frac{1}{S^2}+1\right)^{\frac{3}{2}}S^3R^2}\right) = \mathcal{O}\left(\frac{1}{S^2}\right).
\end{equation*} Thus,
\begin{equation*}
    \left|\frac{\partial}{\partial a}A(X,a,b)\right| = \left|\frac{\partial}{\partial a}S \cdot \frac{\partial}{\partial S}A(X,S)\right| = \left|\frac{-S}{a}\cdot \frac{\partial}{\partial S}A(X,S)\right| = \mathcal{O}\left(\frac{1}{S}\right)
\end{equation*}
because $a > \frac{a^*}{100}$. Similarly, because $b > \frac{b^*}{100}$,
\begin{equation*}
    \left|\frac{\partial}{\partial b}A(X,a,b)\right| = \left|\frac{\partial}{\partial b}S \cdot \frac{\partial}{\partial S}A(X,S)\right| = \left|\frac{S}{b}\cdot \frac{\partial}{\partial S}A(X,S)\right| = \mathcal{O}\left(\frac{1}{S}\right).
\end{equation*}
Therefore, it is reasonable to assume the amplitude remains unchanged for large values of $X$, and we let $A$ be this amplitude. We approximate a predicted wave function $g(t,a,b) = A(X,a,b)f(t - T(X,a,b))$ where $X$ is a fixed large constant by the new function $g_1(t,a,b) = Af(t - T(X,a,b))$ so it suffices to approximate the observed wave function $h(t) = g(t,a^*,b^*)$ by the wave function $h_1(t) = g_1(t,a^*,b^*)$.

As the $W_2$ distance is only defined when both of its inputs have total mass $1$, we normalize all wave data of the form $k(t)$ using the formula \begin{equation*}
    \widetilde{k}(t) = \frac{k(t) + \gamma}{\int_0^\mathcal{T} (k(t) + \gamma)\text{d}t}
\end{equation*} on the interval $[0,\mathcal{T}]$ and $0$ everywhere else. Here, $\gamma$ is a positive constant such that $k(t) + \gamma > 0$ for all $t$.

Eventually, we compute the squared $W_2$ distance between $\widetilde{g}_1(t,a,b)$ and $\widetilde{h}_1(t)$.
We have
\begin{equation}\label{eq:normalized_functions}
    \widetilde{g}_1(t,a,b) = \frac{Af(t-T_{pred}) + \gamma}{AI_0 + \gamma\mathcal{T}}\hspace{2mm}\text{and}\hspace{2mm}
    \widetilde{h}_1(t) = \widetilde{g}_1(t,a^*,b^*) = \frac{Af(t-T_{obs}) + \gamma}{AI_0 + \gamma\mathcal{T}},
\end{equation} where $\int_0^\mathcal{T} f(t)\hspace{1mm}\text{d}t = I_0$, $T_{pred} = T(X,a,b)$, and $T_{obs} = T(X,a^*,b^*)$. We also assume that $\gamma$ is large enough to ensure that $\widetilde{g}_1(t,a,b)$ and $\widetilde{h}_1(t)$ are strictly positive with total mass $1$.

Then, we let \begin{equation*}
    G(t,a,b) = \int_0^t \widetilde{g}_1(y,a,b)\hspace{1mm}\text{d}y \hspace{2mm}\text{and}\hspace{2mm}
    H(t) = \int_0^t \widetilde{h}_1(y)\hspace{1mm}\text{d}y = G(t,a^*,b^*).
\end{equation*}
Since $\widetilde{g}_1(t,a,b)$ and $\widetilde{h}_1(t)$ are both probability distributions, we may compute the squared $W_2$ distance between them using \Cref{eq:W2_formula} to get
\begin{equation*}
    W_2^2(\widetilde{g}_1,\widetilde{h}_1) = \int_0^1 (G^{-1}(s,a,b) - H^{-1}(s))^2\hspace{1mm}\text{d}s
\end{equation*} where $G^{-1}(s,a,b):[0,1]\to[0,\mathcal{T}]$ is the unique function satisfying
$$G(G^{-1}(s,a,b),a,b) = s\text{ and }
G^{-1}(G(t,a,b),a,b) = t$$ for all $s\in[0,1]$ and $t\in[0,\mathcal{T}]$. 

Observe that $\widetilde{g}_1 = \widetilde{g}_1(t,a,b)$ and $G = G(t,a,b)$ can be expressed as functions of time $t$ and the velocity parameters $a,b$, and $G^{-1}(s) = G^{-1}(s,a,b)$ can be expressed as a function of $s$ and $a,b$. It is also possible to express $\widetilde{g}_1 = \widetilde{g}_1(t,T_{pred})$ as a function of $t$ and the predicted traveltime $T_{pred}$ using \Cref{eq:normalized_functions}. Thus, we can alternatively express $G = G(t,T_{pred})$ as a function of $t$ and $T_{pred}$, and $G^{-1}(s) = G^{-1}(s,T_{pred})$ as a function of $s$ and $T_{pred}$. In the proofs of the following claims, we sometimes omit the variables $a,b,$ and $T_{pred}$ in the arguments of $\widetilde{g}_1, G,$ and $G^{-1}$.

We compute $\int_0^1(G^{-1}(s))^2\hspace{1mm}\text{d}s$ to simplify $W_2^2(\widetilde{g}_1,\widetilde{h}_1)$.
\begin{lem}\label{lem:G-inverse_square_integral}
Let $I_0 = \int_0^\mathcal{T}f(t)\hspace{1mm}\mathrm{d}t$, $I_1 = \int_0^\mathcal{T}tf(t)\hspace{1mm}\mathrm{d}t$, and $I_2 = \int_0^\mathcal{T}t^2f(t)\hspace{1mm}\mathrm{d}t$. Then, if $T_{obs}\le T_{pred}\le \mathcal{T} - p_2$,
\begin{equation*}
    \int_0^1(G^{-1}(s,T_{pred}))^2\hspace{1mm}\text{d}s = \frac{1}{I_0 + \frac{\gamma\mathcal{T}}{A}}\left(\frac{\gamma\mathcal{T}^3}{3A}+ I_2 + 2T_{pred}I_1 + T^2_{pred}I_0\right).
\end{equation*}
\end{lem}
\begin{proof}
From the substitution $s = G(t)$, we can compute $\int_0^1(G^{-1}(s))^2\hspace{1mm}\text{d}s$ as 
\begin{equation*}
    \int_0^1(G^{-1}(s))^2\hspace{1mm}\text{d}s = \int_0^\mathcal{T}t^2\widetilde{g}_1(t)\hspace{1mm}\text{d}t = \int_0^\mathcal{T}t^2\frac{f(t-T_{pred}) + \frac{\gamma}{A}}{I_0 + \frac{\gamma\mathcal{T}}{A}}\hspace{1mm}\text{d}t.
\end{equation*} We may write the integral as
\begin{equation*}
    \int_0^\mathcal{T}t^2\frac{f(t-T_{pred}) + \frac{\gamma}{A}}{I_0 + \frac{\gamma\mathcal{T}}{A}}\hspace{1mm}\text{d}t = \frac{1}{I_0 + \frac{\gamma\mathcal{T}}{A}}\left(\frac{\gamma\mathcal{T}^3}{3A}+\int_0^\mathcal{T}t^2{f(t-T_{pred})}\hspace{1mm}\text{d}t\right).
\end{equation*}
However, note that 
$
    \int_0^\mathcal{T}t^2{f(t-T_{pred})}\hspace{1mm}\text{d}t = \int_{-T_{pred}}^{\mathcal{T}-T_{pred}}(t+T_{pred})^2{f(t)}\hspace{1mm}\text{d}t.
$
As $f$ has compact support $[p_1,p_2]\subset [0,\mathcal{T}]$ and $T_{pred} \le \mathcal{T}-p_2$, we may write the integral as $\int_{0}^{\mathcal{T}}(t+T_{pred})^2{f(t)}\hspace{1mm}\text{d}t$. This becomes $I_2 + 2T_{pred}I_1 + T_{pred}^2I_0$, where $I_2 = \int_0^{\mathcal{T}}t^2f(t)\hspace{1mm}\text{d}t, I_1 = \int_0^{\mathcal{T}}tf(t)\hspace{1mm}\text{d}t$, and $I_0 = \int_0^{\mathcal{T}}f(t)\hspace{1mm}\text{d}t$, which proves the lemma.
\end{proof}
Now, observe that the squared $W_2$ distance may be expressed as a function of $T_{pred} = T(X,a,b)$, because $G^{-1}(s,T_{pred})$ is a function of both $s$ and $T_{pred}$. Using this, we may compute the first and second derivatives of $W_2^2(\widetilde{g}_1,\widetilde{h}_1)$ with respect to $T_{pred}$. To compute the first derivative of $W_2^2(\widetilde{g}_1,\widetilde{h}_1)$ with respect to $T_{pred}$, it is enough to find the first derivative of $G^{-1}$ with respect to $T_{pred}$.
\begin{lem}\label{lem:G-inverse_first_derivative_T_pred}
The first derivative of $G^{-1}(s,T_{pred})$ with respect to $T_{pred}$ is
\begin{equation}\label{eq:G-inverse_first_derivative_T_pred}
    \frac{\partial}{\partial T_{pred}}G^{-1}(s, T_{pred}) = \frac{f(G^{-1}(s, T_{pred})-T_{pred})}{f(G^{-1}(s,T_{pred})-T_{pred}) + \frac{\gamma}{A}} = 1 - \frac{\frac{\gamma}{A}}{f(G^{-1}(s,T_{pred})-T_{pred}) + \frac{\gamma}{A}}.
\end{equation}
\end{lem}
\begin{proof}
    Because
    $\int_0^{G^{-1}(s)}\widetilde{g}_1(t)\hspace{1mm}\text{d}t = G(G^{-1}(s)) = s,$
we may apply the Leibniz integral rule to see that
\begin{align*}
    0 = \frac{\partial}{\partial T_{pred}} s
    &= \frac{\partial}{\partial T_{pred}} \int_0^{G^{-1}(s)}\widetilde{g}_1(t)\hspace{1mm}\text{d}t& \\
    &= \widetilde{g}_1(G^{-1}(s))\frac{\partial}{\partial T_{pred}}G^{-1}(s) + \int_0^{G^{-1}(s)}\frac{\partial}{\partial T_{pred}}\widetilde{g}_1(t)\hspace{1mm}\text{d}t.&
\end{align*}
Thus, we have that 
\begin{equation*}
    \frac{\partial}{\partial T_{pred}}G^{-1}(s) = \frac{-\int_0^{G^{-1}(s)}\frac{\partial}{\partial T_{pred}}\widetilde{g}_1(t)\hspace{1mm}\text{d}t}{\widetilde{g}_1(G^{-1}(s))} = \frac{-\left(I_0 + \frac{\gamma\mathcal{T}}{A}\right)\int_0^{G^{-1}(s)}\frac{\partial}{\partial T_{pred}}\widetilde{g}_1(t)\hspace{1mm}\text{d}t}{f(G^{-1}(s)-T_{pred})+\frac{\gamma}{A}}.
\end{equation*}
From \Cref{eq:normalized_functions}, we see that
\begin{equation*}
    -\int_0^{G^{-1}(s)}\frac{\partial}{\partial T_{pred}}\widetilde{g}_1(t)\hspace{1mm}\text{d}t = \frac{\int_0^{G^{-1}(s)}f'(t-T_{pred})\hspace{1mm}\text{d}t}{I_0 + \frac{\gamma\mathcal{T}}{A}}
    =\frac{f(G^{-1}(s)-T_{pred})}{I_0 + \frac{\gamma\mathcal{T}}{A}},
\end{equation*} where we use the fact that $f(-T_{pred}) = 0$. Thus, \Cref{eq:G-inverse_first_derivative_T_pred} holds.
\end{proof}
Now, we use this to compute the first derivative of the squared $W_2$ distance with respect to $T_{pred}$.
\begin{lem}\label{lem:final_Tpred_derivative_W2}
Suppose $T_{obs}\le T_{pred}\le \mathcal{T} - p_2$. The first derivative of $W_2^2(\widetilde{g}_1(t,T_{pred}),\widetilde{h}_1(t))$ with respect to $T_{pred}$ is
\begin{equation*}
    \frac{2}{I_0 + \frac{\gamma\mathcal{T}}{A}} \left(\left(I_0-\frac{\gamma\mathcal{T}}{A}\right)(T_{pred}-T_{obs}) + \frac{\gamma}{A}\int_0^\mathcal{T}\int_{T_{obs}}^{T_{pred}}\frac{\frac{\gamma}{A}}{f(G^{-1}(H(t),y)-y)+\frac{\gamma}{A}}\hspace{1mm}\text{d}y\hspace{1mm}\text{d}t\right).
\end{equation*}
\end{lem}
\begin{proof}
Because $\frac{\partial}{\partial T_{pred}}(H^{-1}(s)) = 0$, 
\begin{equation*}
    \frac{\partial}{\partial T_{pred}}W_2^2(\widetilde{g}_1,\widetilde{h}_1) = 
    \frac{\partial}{\partial T_{pred}}\int_0^1(G^{-1}(s))^2\hspace{1mm}\text{d}s - 2\int_0^1H^{-1}(s)\frac{\partial}{\partial T_{pred}}(G^{-1}(s))\hspace{1mm}\text{d}s,
\end{equation*}
using the Leibniz integral rule. By \Cref{lem:G-inverse_square_integral} and \Cref{lem:G-inverse_first_derivative_T_pred} we simplify
\begin{equation*}
    \frac{\partial}{\partial T_{pred}}W_2^2(\widetilde{g}_1,\widetilde{h}_1) = \frac{2T_{pred}I_0 + 2I_1}{I_0 + \frac{\gamma\mathcal{T}}{A}} - 2\int_0^1H^{-1}(s) - \frac{\frac{\gamma}{A}H^{-1}(s)}{f(G^{-1}(s)-T_{pred}) + \frac{\gamma}{A}}\hspace{1mm}\text{d}s.
\end{equation*} From the change of variables $s = H(t)$, we see that $\int_0^1H^{-1}(s)\hspace{1mm}\text{d}s = \int_0^\mathcal{T}t\widetilde{h}_1(t)\hspace{1mm}\text{d}t$, which can be computed as 
\begin{align*} \int_0^\mathcal{T}
    \frac{tf(t-T_{obs}) + \frac{\gamma t}{A}}{I_0 + \frac{\gamma\mathcal{T}}{A}}\hspace{1mm}\text{d}t &= \frac{1}{I_0 + \frac{\gamma\mathcal{T}}{A}}\left( \frac{\gamma\mathcal{T}^2}{2A} + \int_{-T_{obs}}^{\mathcal{T}-T_{obs}}
    (t+T_{obs})f(t)\hspace{1mm}\text{d}t\right)&
    \\ &=
    \frac{\frac{\gamma\mathcal{T}^2}{2A}+I_1 + T_{obs}I_0}{I_0 + \frac{\gamma\mathcal{T}}{A}}.&
\end{align*}
This means
\begin{equation*}
 \frac{\partial}{\partial T_{pred}}W_2^2(\widetilde{g}_1,\widetilde{h}_1) = \frac{2}{I_0 + \frac{\gamma\mathcal{T}}{A}} \left(I_0(T_{pred}-T_{obs}) + \frac{\gamma}{A}\int_0^1 \frac{H^{-1}(s)}{\widetilde{g}_1(G^{-1}(s))} \hspace{1mm}\text{d}s - \frac{\gamma\mathcal{T}^2}{2A}\right).
\end{equation*}
Furthermore, observe that $\int_0^1 \frac{H^{-1}(s)}{\widetilde{g}_1(G^{-1}(s))} \hspace{1mm}\text{d}s = \int_0^\mathcal{T} H^{-1}(G(t)) \hspace{1mm}\text{d}t$ from the substitution $s = G(t)$. 
Because $H^{-1}(G(t))$ is the inverse function of $G^{-1}(H(t))$, we have that $\int_0^\mathcal{T} H^{-1}(G(t)) \hspace{1mm}\text{d}t  = \mathcal{T}^2 - \int_0^\mathcal{T} G^{-1}(H(t)) \hspace{1mm}\text{d}t$. Thus,
\begin{equation}\label{eq:final_W2_derivative_T_pred}
    \frac{\partial}{\partial T_{pred}}W_2^2(\widetilde{g}_1,\widetilde{h}_1) = 
    \frac{2}{I_0 + \frac{\gamma\mathcal{T}}{A}} \left(I_0(T_{pred}-T_{obs}) +  \frac{\gamma\mathcal{T}^2}{2A} - \frac{\gamma}{A}\int_0^\mathcal{T} G^{-1}(H(t)) \hspace{1mm}\text{d}t\right).
\end{equation}
Observe that  $G^{-1}(H(t),T_{obs}) = t$. From \Cref{eq:G-inverse_first_derivative_T_pred}, we have that
\begin{align*}
    G^{-1}(H(t),T_{pred}) &= t + \int_{T_{obs}}^{T_{pred}}\frac{\partial}{\partial y}G^{-1}(H(t),y)\hspace{1mm}\text{d}y& \\&= t + \int_{T_{obs}}^{T_{pred}}\left(1-\frac{\frac{\gamma}{A}}{f(G^{-1}(H(t),y)-y)+\frac{\gamma}{A}}\right)\hspace{1mm}\text{d}y.&
\end{align*} Substituting this expression for $G^{-1}(H(t))$ into \Cref{eq:final_W2_derivative_T_pred} gives
\begin{equation*}
    \frac{2}{I_0 + \frac{\gamma\mathcal{T}}{A}} \left(\left(I_0-\frac{\gamma\mathcal{T}}{A}\right)(T_{pred}-T_{obs}) + \frac{\gamma}{A}\int_0^\mathcal{T}\int_{T_{obs}}^{T_{pred}}\frac{\frac{\gamma}{A}}{f(G^{-1}(H(t),y)-y)+\frac{\gamma}{A}}\hspace{1mm}\text{d}y\hspace{1mm}\text{d}t\right)
\end{equation*} as the value of $\frac{\partial}{\partial T_{pred}}W_2^2(\widetilde{g}_1,\widetilde{h}_1)$.
\end{proof}
Expressing $\frac{\partial}{\partial T_{pred}}W_2^2(\widetilde{g}_1,\widetilde{h}_1)$ in this form allows us to prove the convexity of $W_2^2(\widetilde{g}_1,\widetilde{h}_1)$ with respect to $a,b$ subject to a restriction on the source function $f$. Here, we assume that $f$ reaches both positive and negative values.
\begin{thm}\label{thm:W2_amp_convex}
Let $q = \frac{\frac{\gamma}{A}}{\sup f + \frac{\gamma}{A}}$ and
suppose that $I_0 \ge (1-q)\frac{\gamma\mathcal{T}}{A}$. Let $a^*$ and $b^*$ be positive constants such that $\frac{b^*X}{2a^*} \ge S_0$ and let
\begin{equation*}V = \{(a,b)\in\R^2 : a > \frac{a^*}{100}, \frac{bX}{2a}\ge \frac{b^*X}{2a^*}, T(X,a^*,b^*)\le T(X,a,b)\le \mathcal{T} - p_2\}.
\end{equation*}
Then, $V$ is nonempty and
$W_2^2(\widetilde{g}_1(t,a,b),\widetilde{h}_1(t))$ is jointly convex in $(a,b)\in V$.
\end{thm}
Before beginning the proof of \Cref{thm:W2_amp_convex}, we remark that the function $f(t) = (t-15)^2 - A$ with domain $[0,\mathcal{T}]$, where $\mathcal{T} = 30, \gamma = 2A$, and $0 < A \le 73$, satisfies the condition $I_0 \ge (1-q)\frac{\gamma\mathcal{T}}{A}$ and reaches both positive and negative values. 
To see this, observe that $I_0 = 2250 - 30A \ge 60$ while $(1-q)\frac{\gamma\mathcal{T}}{A} = \frac{60\sup f}{\sup f + 2} < 60$.
Hence, our assumption that $I_0 \ge (1-q)\frac{\gamma\mathcal{T}}{A}$ is a reasonable assumption to make.
\begin{proof} 

First, the point $(a^*, b^*)$ is clearly in $V$, so $V$ is nonempty. Now, we prove that $W_2^2(\widetilde{g}_1(t,a,b),\widetilde{h}_1(t))$ is jointly convex in $(a,b)\in V$. Because $\frac{b^*X}{2a^*} \ge S_0$, $V$ is a subset of $U$, which is defined in \Cref{thm:2d_time_shift}.
This implies that $T(X,a,b)$ is jointly convex in $a,b$ over the region $V$ by \Cref{lem_linear_depth}.
Thus, it is enough to show that $W_2^2(\widetilde{g}_1,\widetilde{h}_1)$ is convex and nondecreasing in $T_{pred}$ in the interval $[T_{obs},\mathcal{T}-p_2]$.
We first claim that $W_2^2(\widetilde{g}_1,\widetilde{h}_1)$ is nondecreasing in $T_{pred}$ in this interval.
By \Cref{lem:final_Tpred_derivative_W2}, we see that
\begin{align*}
    \frac{\partial}{\partial T_{pred}}W_2^2(\widetilde{g}_1,\widetilde{h}_1)
    &\ge \frac{2}{I_0 + \frac{\gamma\mathcal{T}}{A}} \left(\left(I_0-\frac{\gamma\mathcal{T}}{A}\right)(T_{pred}-T_{obs}) + \frac{\gamma}{A}\int_0^\mathcal{T}\int_{T_{obs}}^{T_{pred}}q\hspace{1mm}\text{d}y\hspace{1mm}\text{d}t\right)&\\
    &=\frac{2}{I_0 + \frac{\gamma\mathcal{T}}{A}}\left(I_0+(q-1)\frac{\gamma\mathcal{T}}{A}\right)(T_{pred}-T_{obs})&
\end{align*} because $T_{pred} \ge T_{obs}$. Furthermore, because $I_0 \ge (1-q)\frac{\gamma\mathcal{T}}{A}$, every factor is nonnegative, implying that $\frac{\partial}{\partial T_{pred}}W_2^2(\widetilde{g}_1,\widetilde{h}_1)$ is nonnegative. 
Next, we claim that $W_2^2(\widetilde{g}_1,\widetilde{h}_1)$ is convex in $T_{pred}$ in the interval $[T_{obs},\mathcal{T}-p_2]$.
Using \Cref{eq:final_W2_derivative_T_pred} and \Cref{lem:G-inverse_first_derivative_T_pred}, we see that 
\begin{align*}
\frac{\partial^2}{\partial T_{pred}^2}W_2^2(\widetilde{g}_1,\widetilde{h}_1)
    &= \frac{\partial}{\partial T_{pred}}\left(
    \frac{2}{I_0 + \frac{\gamma\mathcal{T}}{A}} \left(I_0(T_{pred}-T_{obs}) +  \frac{\gamma\mathcal{T}^2}{2A} - \frac{\gamma}{A}\int_0^\mathcal{T} G^{-1}(H(t)) \hspace{1mm}\text{d}t\right)\right)&\\
    &= \frac{2}{I_0 + \frac{\gamma\mathcal{T}}{A}} \left(I_0 - \frac{\gamma}{A}\int_0^\mathcal{T}\frac{\partial}{\partial T_{pred}} G^{-1}(H(t)) \hspace{1mm}\text{d}t\right)&\\
    &= \frac{2}{I_0 + \frac{\gamma\mathcal{T}}{A}} \left(I_0 - \frac{\gamma}{A}\int_0^\mathcal{T}1 - \frac{\frac{\gamma}{A}}{f(G^{-1}(H(t)) - T_{pred}) + \frac{\gamma}{A}} \hspace{1mm}\text{d}t\right)&\\
    &\ge \frac{2}{I_0 + \frac{\gamma\mathcal{T}}{A}}\left(I_0 + (q-1) \frac{\gamma\mathcal{T}}{A}\right).&
\end{align*} Since $I_0 \ge (1-q)\frac{\gamma\mathcal{T}}{A}$, we see that $\frac{\partial^2}{\partial T_{pred}^2} W_2^2(\widetilde{g}_1,\widetilde{h}_1)$ is nonnegative in the interval $[T_{obs},\mathcal{T}-p_2]$. Thus, $W_2^2(\widetilde{g}_1,\widetilde{h}_1)$ is convex and nondecreasing in $T_{pred}$ in the interval $[T_{obs},\mathcal{T}-p_2]$. Since $T_{pred}$ is jointly convex in $(a,b)\in V$, $W_2^2(\widetilde{g}_1,\widetilde{h}_1)$ is also jointly convex in $(a,b)\in V$ by \Cref{lem:compconvex}.
\end{proof}
\begin{rem}
Through the Mean Value Theorem, \Cref{thm:W2_amp_convex} can be reformulated as a result with a condition involving an upper bound on $\frac{\text{d}f}{\text{d}t}$. This reformulation illustrates how the convexity of $W_2^2(\widetilde{g}_1,\widetilde{h}_1)$ is influenced by the frequency of the source function.
\end{rem}
\begin{rem}
The requirements needed to apply \Cref{lem:compconvex} are not satisfied for $T_{pred}\in(0,T_{obs}]$. To see this, observe that the squared $W_2$ distance is nonnegative everywhere and $0$ when $T_{pred} = T_{obs}$. Thus, it is impossible for the squared $W_2$ distance to be nondecreasing in $T_{pred}\in(0,T_{obs}]$, although it may be convex in this interval.
\end{rem}
\section{Numerical Results}
\begin{figure}
\begin{subfigure}{.33\textwidth}
  \centering
  \includegraphics[scale = 0.35]{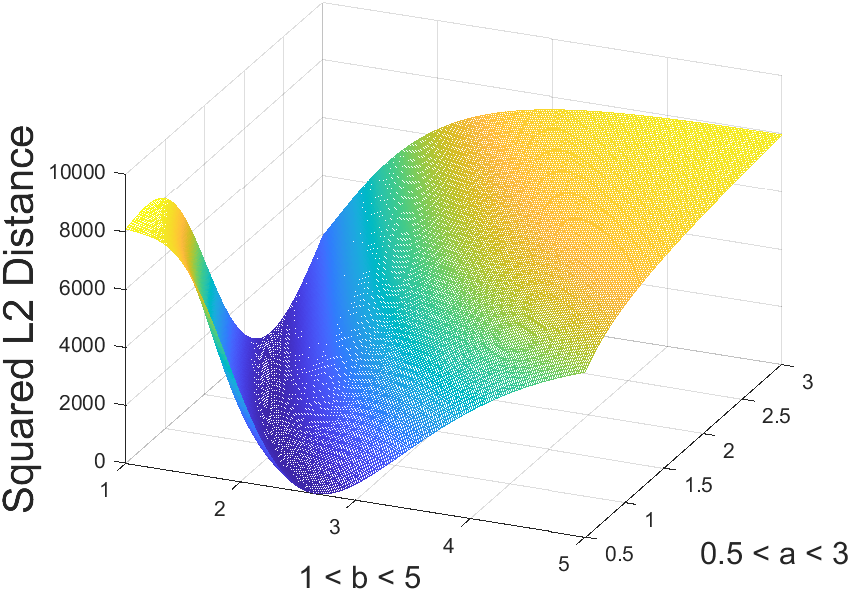}
  \caption{$\alpha = 2$}
  \label{fig:L2 Amplitude Plot 2}
\end{subfigure}%
\begin{subfigure}{.33\textwidth}
  \centering
  \includegraphics[scale = 0.35]{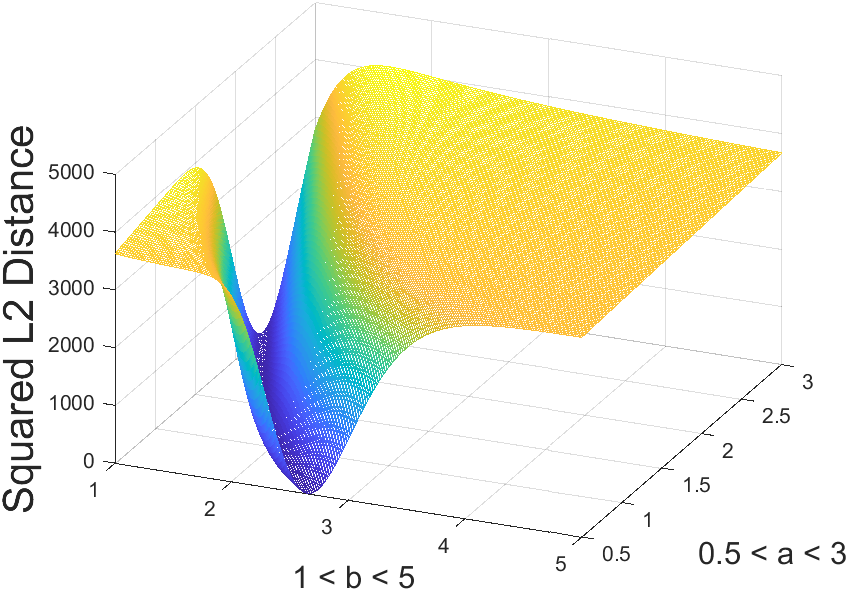}
  \caption{$\alpha = 10$}
  \label{fig:L2 Amplitude Plot 10}
\end{subfigure}%
\begin{subfigure}{.33\textwidth}
  \centering
  \includegraphics[scale = 0.35]{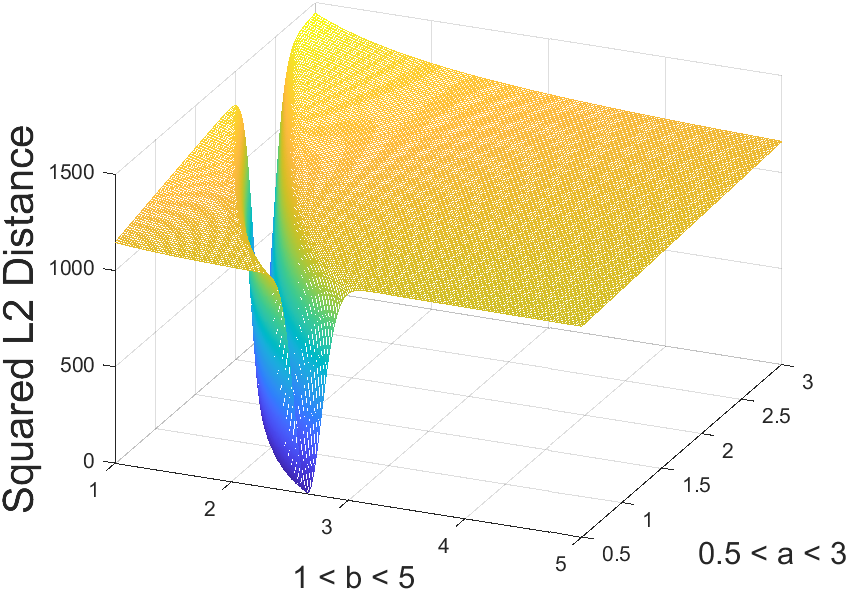}
  \caption{$\alpha = 100$}
  \label{fig:L2 Amplitude Plot 100}
\end{subfigure}%
    \caption{Plots of
    $||g-h||_2^2$, where there is only one receiver at $X=10$.}
    \label{fig:L2 Distance Plot}
\end{figure}
\begin{figure}

\begin{subfigure}{.33\textwidth}
  \centering
  \includegraphics[scale = 0.35]{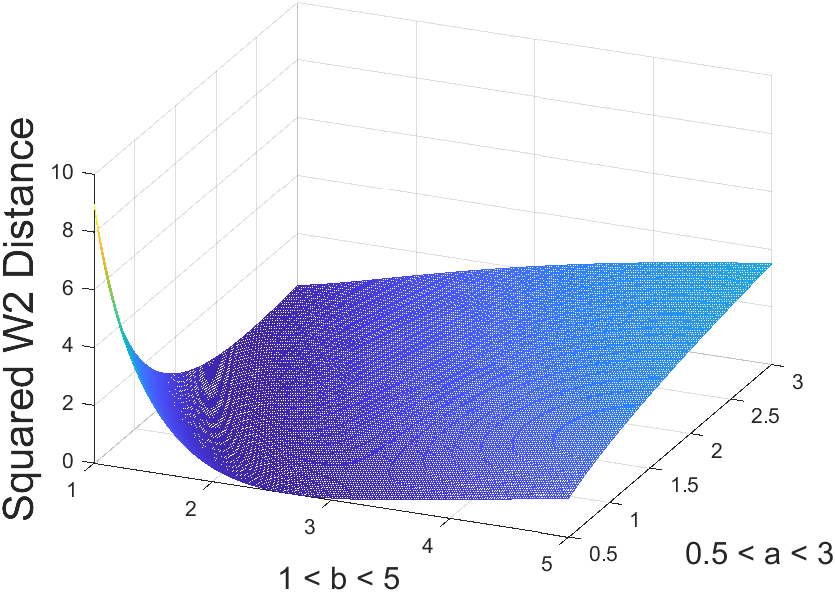}
  \caption{$\alpha = 2$}
  \label{fig:W2 Amplitude Plot 2}
\end{subfigure}%
\begin{subfigure}{.33\textwidth}
  \centering
  \includegraphics[scale = 0.35]{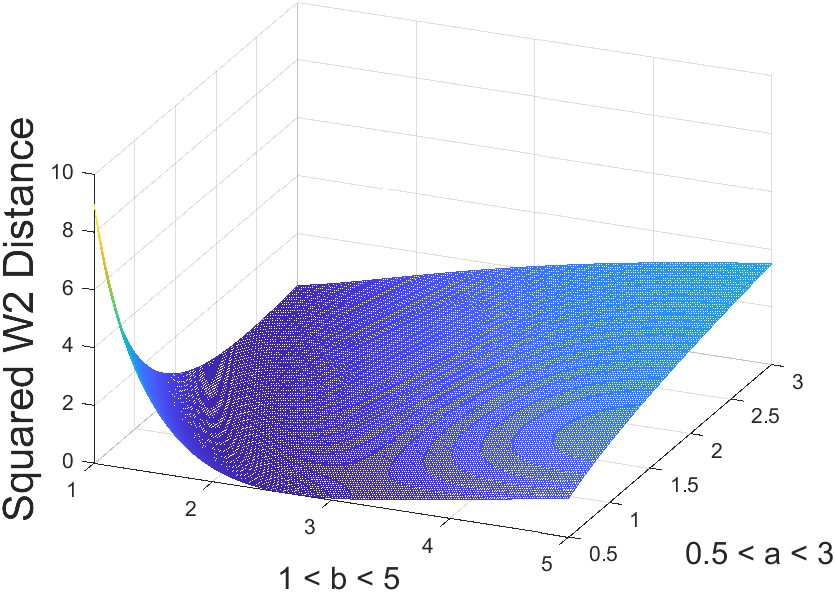}
  \caption{$\alpha = 10$}
  \label{fig:W2 Amplitude Plot 10}
\end{subfigure}%
\begin{subfigure}{.33\textwidth}
  \centering
  \includegraphics[scale = 0.35]{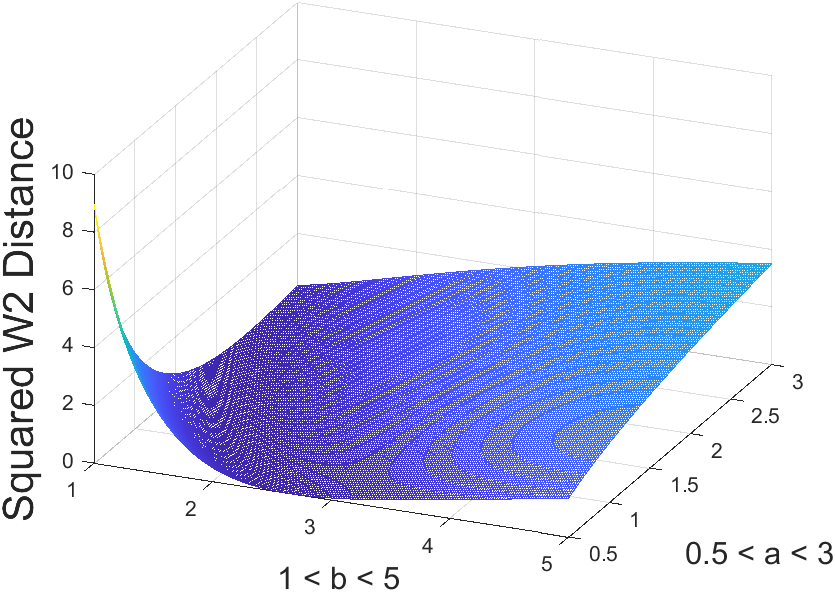}
  \caption{$\alpha = 100$}
  \label{fig:W2 Amplitude Plot 100}
\end{subfigure}%
    \caption{Plots of
    $W_2^2(\widetilde{g}(X,t),\widetilde{h}(X,t))$, where there is only one receiver at $X=10$.}
    \label{fig:Single W2 Convexity Plot}
\end{figure}

\begin{figure}

\begin{subfigure}{.33\textwidth}
  \centering
  \includegraphics[scale = 0.35]{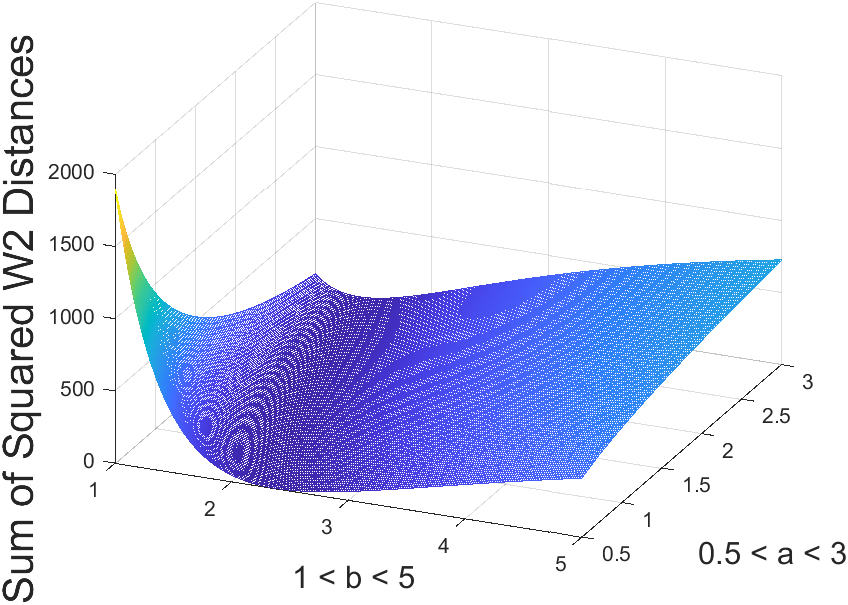}
  \caption{$\alpha = 2$}
  \label{fig:Multiple W2 Amplitude Plot 2}
\end{subfigure}%
\begin{subfigure}{.33\textwidth}
  \centering
  \includegraphics[scale = 0.35]{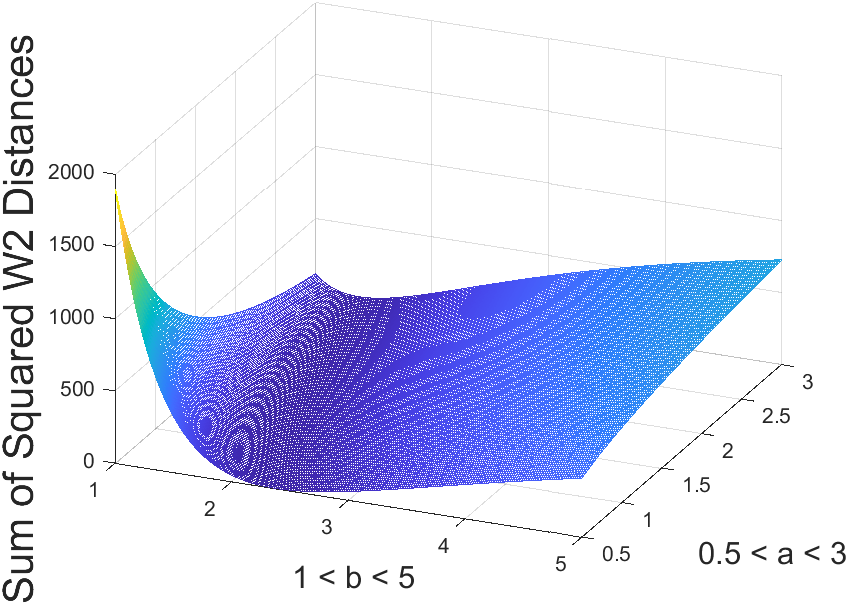}
  \caption{$\alpha = 10$}
  \label{fig:Multiple W2 Amplitude Plot 10}
\end{subfigure}%
\begin{subfigure}{.33\textwidth}
  \centering
  \includegraphics[scale = 0.35]{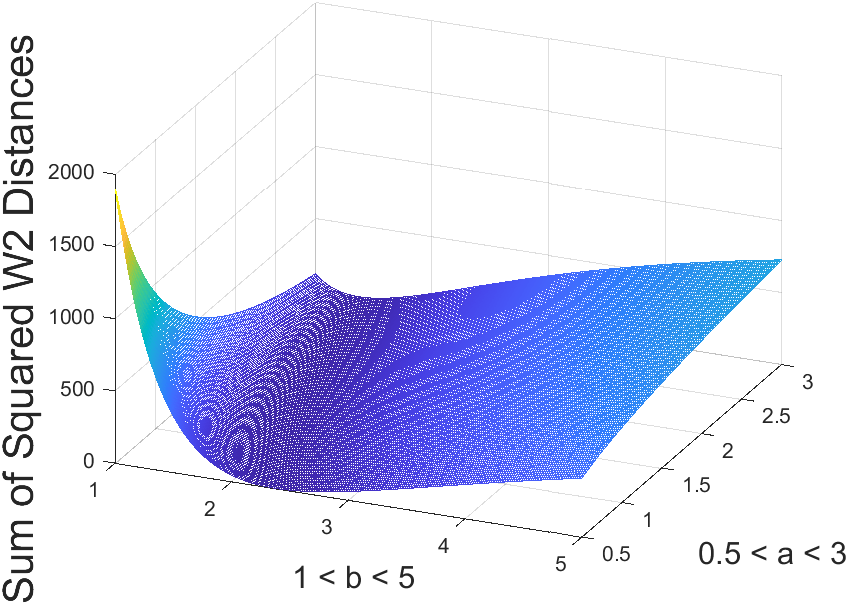}
  \caption{$\alpha = 100$}
  \label{fig:Multiple W2 Amplitude Plot 100}
\end{subfigure}%
    \caption{Plots of
    $\sum_{r} W_2^2(\widetilde{g}(X_r,t),\widetilde{h}(X_r,t))$, where the receiver locations $X_r$ range from $10$ to $100$, inclusive.}
    \label{fig:Multiple W2 Convexity Comparison}
\end{figure}
We continue with the velocity model studied in \Cref{subsection:2d_velocity_model}.
We compute both the $L^2$ and $W_2$ distance and compare the convexity of the two objective functions in $(a,b)$, as shown in \Cref{fig:L2 Distance Plot,fig:Single W2 Convexity Plot}. The source wave function is of the form $f(t) = e^{-\alpha(t-5)^2}$ and we consider $\alpha\in\{2,10,100\}$. Here, $(a^*,b^*) = (1,2)$ and $\gamma = 10^{-6}$. We can consider the observed data to be of the form $A(X,a,b)f(t-T(X,a,b))$ because the velocity is a continuous function of the depth. To ensure that the wave data is compactly supported, we use the time range $[0,50]$ to compute the squared $W_2$ distance.

While the squared $W_2$ distance appears to be mostly convex for $\alpha\in\{2,10,100\}$, the squared $L^2$ norm is certainly nonconvex. The plot of the squared $L^2$ norm has large flat regions with a steep incline closer to where the $L^2$ norm is minimized, as shown in \Cref{fig:L2 Amplitude Plot 2,fig:L2 Amplitude Plot 10,fig:L2 Amplitude Plot 100}. Although the squared $L^2$ norm is nonconvex, by decreasing the value of $\alpha$ we increase the size of the convex region around $(a^*,b^*)$ of the squared $L^2$ norm. As the graph of the source function becomes sharper (\Cref{fig:Source function plots}), so does the graph of the squared $L^2$ norm. The squared $W_2$ distance, on the other hand, is relatively flat throughout the entire domain and does not have a steep incline closer to the minimum, as shown in \Cref{fig:W2 Amplitude Plot 2,fig:W2 Amplitude Plot 10,fig:W2 Amplitude Plot 100}. Thus, the squared $W_2$ distance should be convex on a much larger region containing the minimum. The squared $W_2$ distance is also highly insensitive to the choice of source function, and this suggests that the squared $W_2$ distance can be used to solve various seismic inversion problems, in contrast with the squared $L^2$ norm.

In addition, we plot the sum of the squared $W_2$ distance taken over multiple receiver locations in \Cref{fig:Multiple W2 Convexity Comparison}. The summation of the squared $W_2$ distance over multiple receiver locations is highly convex regardless of $\alpha$, as seen in \Cref{fig:Multiple W2 Amplitude Plot 2,fig:Multiple W2 Amplitude Plot 10,fig:Multiple W2 Amplitude Plot 100}. Furthermore, the summation of the squared $W_2$ distance over multiple receiver locations is also very close to the summation of $(T_{pred} - T_{obs})^2$ over multiple receiver locations, and appears to be convex in $(a,b)$ for $(a,b)$ closer to the origin. Thus, as an initial guess for $(a,b)$, it appears to be better to choose points $(a,b)$ which are very close to the origin. This is equivalent to choosing points $(a,b)$ such that the predicted travel time is large.

\subsection{Frequency Analysis}
\begin{figure}
\begin{subfigure}{.5\textwidth}
  \centering
  \includegraphics[scale = 0.6]{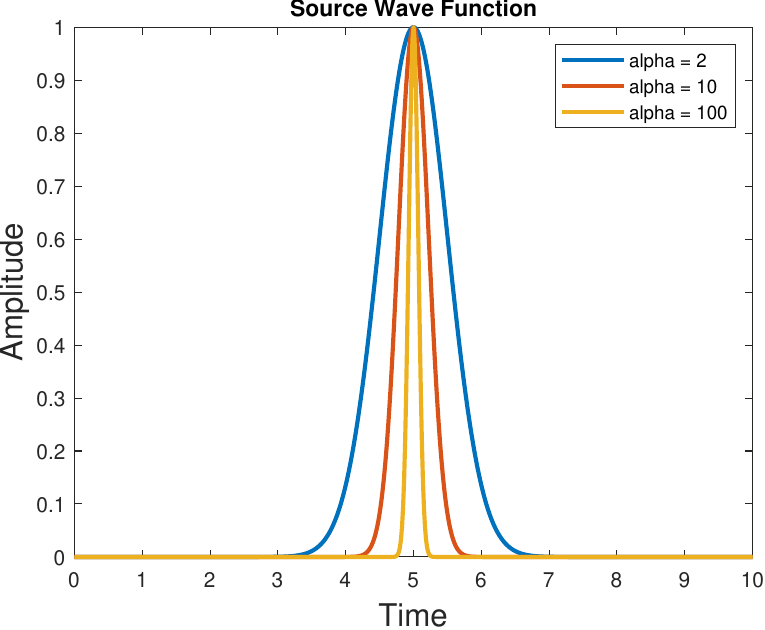}
  \caption{The source functions are of the form $e^{-\alpha(t-5)^2}$.}
  \label{fig:Source function plots}
\end{subfigure}%
\begin{subfigure}{.5\textwidth}
    \centering
    \includegraphics[scale=0.6]{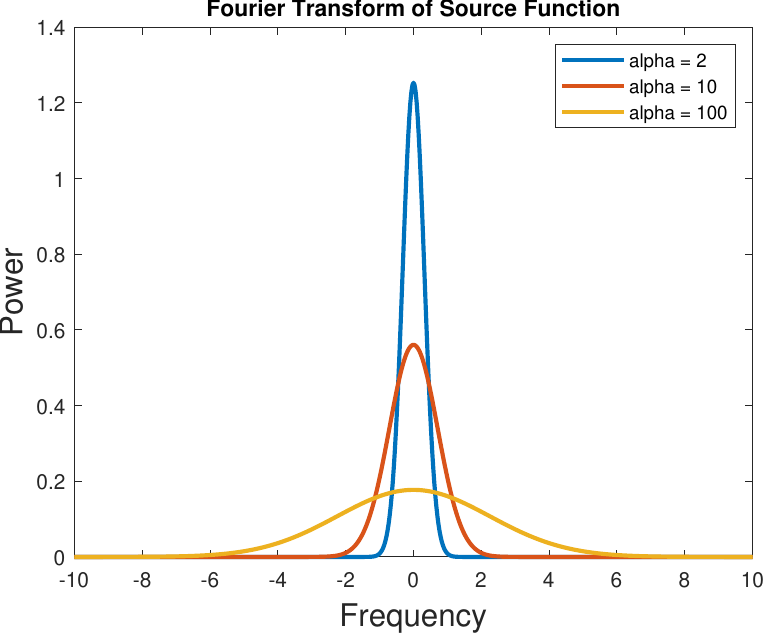}
    \caption{The Fourier transforms of the source functions.}
    \label{fig:Fourier combined plot}
\end{subfigure}%
    \caption{The Fourier transform of the source function depends on the value of $\alpha$. As $\alpha$ increases, the average absolute value of the frequency also increases.}
    \label{fig:Fourier vs. Source plot}
\end{figure}
We observe that the changes in the squared $L^2$ or $W_2$ distances between $g(t,a,b)$ and $h(t)$ depends solely on the value of $\alpha$, which in turn affects the frequency of the source function. The graph of the source function in the frequency domain can be derived by taking a Fourier transform. Letting \begin{equation*}
    \widehat{f}(k) = \int_{-\infty}^\infty e^{-2\pi ikt}f(t)\hspace{1mm}\text{d}t = \int_{-\infty}^\infty e^{-2\pi ikt}e^{-\alpha(t-5)^2}\hspace{1mm}\text{d}t,
\end{equation*} we get that $\widehat{f}(k) = \sqrt{\frac{\pi}{\alpha}}e^{-\frac{(\pi k)^2}{\alpha}-10\pi ik}$, and the power at a frequency $k$ is given by the magnitude, which is $|\widehat{f}(k)| = \sqrt{\frac{\pi}{\alpha}}e^{-\frac{(\pi k)^2}{\alpha}}$. The graph of $|\widehat{f}(k)|$ is centered at $k = 0$ regardless of the value of $\alpha$, but as the value of $\alpha$ increases, the plot of $|\widehat{f}(k)|$ becomes wider, as shown in \Cref{fig:Fourier combined plot}. In other words, if $|k_1| > |k_2|$ where $k_1$ and $k_2$ are fixed frequencies, then $\frac{|\widehat{f}(k_1)|}{|\widehat{f}(k_2)|}$ increases as $\alpha$ increases. Thus, as the value of $\alpha$ increases, the average absolute value of the frequency of the source function also increases. For large values of $\alpha$, the plot of the squared $L^2$ norm also has the properties of higher frequency data, as the plot is very sharp close to $(a^*,b^*) = (1,2)$. The high--frequency data seen in \Cref{fig:L2 Amplitude Plot 100} is explained by the squared $L^2$ norm weighting low--frequency and high--frequency terms equally, by the Plancherel Theorem. On the other hand, the plot of the squared $W_2$ distance is virtually unchanged as $\alpha$ increases, suggesting that the frequency of the the squared $W_2$ distance is highly insensitive to the frequency of the source function.

The relationship between the $W_2$ distance between two functions $g$ and $h$ and a weighted $\dot{\mathcal{H}}^{-1}$ distance between them helps provide an explanation for these observations regarding frequency~\cite{engquist2020quadratic, villani2003topics}. We define the space $\dot{\mathcal{H}}^1(\R^d)$ through the seminorm \begin{equation*}
    ||f||^2_{\dot{\mathcal{H}}^1(\R^d)} = \int_{\R^d}|k|^2|\widehat{f}(k)|^2\hspace{1mm}\text{d}k 
\end{equation*} and the space $\dot{\mathcal{H}}^{-1}(\R^d)$ is defined as the dual of ${\dot{\mathcal{H}}}^1(\R^d)$ through the norm \begin{equation*}
    ||f||_{\dot{\mathcal{H}}^{-1}(\R^d)} = \sup\{|\langle \digamma,f\rangle_{L^2}|:||\digamma||_{\dot{\mathcal{H}}^1} \le 1\}.
\end{equation*} It is known~\cite{villani2003topics} that the $W_2$ distance is asymptotically equivalent to the $\dot{\mathcal{H}}^{-1}$ norm, which weights terms of lower frequency over terms with higher frequency. Specifically, if $\mu$ is a probability measure and $\text{d}\pi$ is an infinitesimal perturbation with zero total mass, then $W_2(\mu, \mu + \text{d}\pi) = ||\text{d}\pi||_{\dot{\mathcal{H}}^{-1}_{(\text{d}\mu)}} + o(\text{d}\pi)$ \cite{engquist2020quadratic}.
While the objective functions in \Cref{fig:Single W2 Convexity Plot,fig:Multiple W2 Convexity Comparison} are not globally convex, the relationship between the squared $W_2$ distance and the squared ${\dot{\mathcal{H}}^{-1}}$ metric offers an explanation for the smoothness of the plots in these figures, which display properties of low--frequency data. 
\subsection{Optimal Transport for Non--probability Measures}
In general, the wave data tends to alternate between positive and negative values, and the total integral of the observed or predicted wave function does not have to be $1$. Thus, we cannot immediately use the squared $W_2$ distance as our objective function, because it is only defined on probability distributions. The current approach to normalizing the wave data requires two steps: first, transform the wave data to a nonnegative function, and second, divide by the total mass~\cite{engquist2013application}. This ensures that the normalized wave data
satisfies the positivity and total mass requirements. Although there are several possible ways to complete the first step, the known methods of doing this have their own drawbacks.

If the source function $f$ is positive, the first step becomes unnecessary. The requirement in \Cref{thm:2d_time_shift} is not very strong, suggesting that when $f$ is positive, the squared $W_2$ distance is suitable as an objective function.
However, this method does not generalize well to source functions that alternate between negative and positive values. 
In this case, we complete the
first step by initially replacing an alternating function $k$ by $k + \gamma$, where $\inf k + \gamma > 0$.
Then, we divide by the total mass of $k + \gamma$ in the interval $[0,\mathcal{T}]$. This method of normalization takes into account the wave amplitude as well. However, the convexity of the squared $W_2$ distance, in this case, is not as general as with the previous method. Further restrictions on the source $f$ are necessary, as shown by the requirements in \Cref{thm:W2_amp_convex}. This suggests that the squared $W_2$ distance is suitable as an objective function when the normalization constant $\gamma$ is sufficiently close to $0$, or equivalently, when $\inf f$ is sufficiently close to $0$.

\section{Conclusions}
In this paper, we study the convexity of full--waveform inversion using the squared $W_2$ distance as an objective function with respect to the velocity model parameter. We show that the squared $W_2$ distance is a suitable objective function for multiple velocity models when the received signal is nonnegative. Next, we show that the squared $W_2$ distance is suitable in some cases, in a two--dimensional velocity model where the received signal alternates between positive and negative values. We review the smoothing property of the squared $W_2$ distance by its relation to the squared $\dot{\mathcal{H}}^{-1}$ distance, and contrast this with the sharpness of the squared $L^2$ norm, which is very sensitive to high--frequency signals. We also discuss the drawbacks of the normalization methods used in this paper. A natural direction for future research is to generalize the $W_2$ distance to compare functions alternating between positive and negative values.
\section{Acknowledgements}
Firstly, the author would like to thank Dr.\ Yunan Yang for her mentorship and guidance during this project. The author thanks Dr.\ Tanya Khovanova and Boya Song for proofreading this paper and for providing feedback. Finally, the author is thankful to the PRIMES--USA program for making this research project possible. This work is supported in part by the National Science Foundation through grant DMS--1913129.

\bibliographystyle{plain} 
\bibliography{reference}
\end{document}